\documentclass[final,3p,times]{elsarticle}

\usepackage{amscd,bbm,amsmath,amsthm,graphicx,amssymb,latexsym,dsfont,amsfonts,paralist,epsfig,graphics,exscale,ulem}
\usepackage{mathrsfs}
\usepackage[colorlinks=true]{hyperref}
\usepackage{color,soul}

\newcommand{\bK}{\boldsymbol{K}}
\newcommand{\bsigma}{\boldsymbol{\sigma}}
\newcommand{\bi}{\boldsymbol{i}}
\newcommand{\bj}{\boldsymbol{j}}
\newcommand{\bx}{\boldsymbol{x}}
\newcommand{\by}{\boldsymbol{y}}
\newcommand{\bvarphi}{\boldsymbol{\varphi}}

\newtheorem{theorem}{Theorem}[section]

\newtheorem{cor}[theorem]{Corollary}
\newtheorem{prop}[theorem]{Proposition}
\newtheorem{lem}[theorem]{Lemma}

\newtheorem*{BW}{Berger-Wang's formula}
\newtheorem*{E}{Elsner's reduction lemma}
\theoremstyle{definition}
\newtheorem{defn}[theorem]{Definition}
\newtheorem*{OP}{Open problem}
\newtheorem{que}{Question}
\newtheorem{example}[theorem]{Example}

\newtheorem{remark}[theorem]{Remark}

\numberwithin{equation}{section}

\newcommand{\EQ}{\begin{equation}\begin{array}{lllllllll}}
\newcommand{\EE}{\end{array}\end{equation}}
\newcommand{\MT}{\left[ \begin{array}{ccccccccc}}
\newcommand{\EM}{\end{array}\right]}

\newcommand{\eq}{\begin{equation}\begin{array}{lclllllllllllllll}}
\newcommand{\ee}{\end{array}\end{equation}}
\newcommand{\bmt}{\left[ \begin{array}{ccccccccc}}
\newcommand{\emt}{\end{array}\right]}
\newcommand{\bea}{\begin{eqnarray}}
\newcommand{\eea}{\end{eqnarray}}
\newcommand{\bean}{\begin{eqnarray*}}
\newcommand{\eean}{\end{eqnarray*}}

\journal{XXX}
\begin{document}

\begin{frontmatter}


\title{Chaotic Characteristic of Discrete-time Linear Inclusion Dynamical Systems\tnoteref{label1}}

\tnotetext[label1]{Dai was supported by NSF of
China (Grant No.~11071112); T.~Huang was supported by National Priority Research Project NPRP 4-451-2-168 funded by Qatar Research Fund; Y.~Huang was supported partly by NSF of China (Grant No.~11071263) and the NSF of Guangdong Province; and Xiao in part by NSF 0605181 and 1021203 of the United States.}%

 \author[label2]{Xiongping Dai}
 \address[label2]{Department of Mathematics, Nanjing University, Nanjing 210093, People's Republic of China}
 \ead{xpdai@nju.edu.cn}

 \author[label3]{Tingwen Huang}
 \address[label3]{Texas A$\&$M University at Qatar, c/o Qatar Foundation, P.O. Box 5825, Doha, Qatar}
 \ead{tingwen.huang@qatar.tamu.edu}

 \author[label4]{Yu Huang}
 \address[label4]{Department of Mathematics, Zhongshan (Sun Yat-Sen) University, Guangzhou 510275, People's Republic of China}
 \ead{stshyu@mail.sysu.edu.cn}

 \author[label5]{Mingqing Xiao}
\address[label5]{Department of Mathematics, Southern Illinois University, Carbondale, IL 62901-4408, USA}
 \ead{mxiao@math.siu.edu}

\begin{abstract}
Given $K$ real $d$-by-$d$ nonsingular matrices $S_1,\dotsc,S_K$, by extending the well-known Li-Yorke chaotic description of a deterministic nonlinear dynamical system to a discrete-time linear inclusion dynamical system: $x_n\in\left\{S_kx_{n-1}\right\}_{1\le k\le K}$ with $x_0\in\mathbb{R}^d$ and $n\ge1$,
we study the chaotic characteristic of the state trajectory $(x_n(x_0,\sigma))_{n\ge1}$, governed by a switching law $\sigma\colon\mathbb{N}\rightarrow\{1,\dotsc,K\}$, for any initial states $x_0\in\mathbb{R}^d$. Two sufficient conditions are given so that for a ``large'' subset of the space
of all possible switching laws $\sigma$, we have the sharp infinite oscillation as follows:
\begin{equation*}
\liminf\limits_{n\to+\infty}\|x_n(x_0,\sigma)\|=0\quad\textrm{and}\quad
\limsup\limits_{n\to+\infty}\|x_n(x_0,\sigma)\|=+\infty\quad
\forall x_0\in\mathbb{R}^d\setminus\{0\}.
\end{equation*}
This implies that there coexists at least one positive, one zero and one negative Lyapunov exponents and that the trajectories $(x_n(x_0,\sigma))_{n\ge1}$ are extremely sensitive to the initial states $x_0\in\mathbb{R}^d$.
We also show that a periodically stable linear inclusion system, which may be unbounded, does not have any
such chaotic behavior.
\end{abstract}
\begin{keyword}
Linear inclusion system\sep fiber-chaotic switching law\sep discontinuity of Lyapunov exponents\sep periodically stable inclusion system.

\medskip
\MSC[2010] 93C30\sep 37A30\sep 15B52
\end{keyword}
\end{frontmatter}

\section{Introduction}\label{sec1}
Chaos is not only an interesting but also an important subject in the theory of dynamical systems. It is well known that many nonlinear systems can exhibit ``chaotic'' behavior that is sensitive to initial conditions. Intuitively small perturbations in initial conditions (such as those due to rounding errors in numerical computation) yield widely diverging outcomes, rendering long-term prediction impossible in general. Even if a dynamical system is deterministic, i.e., their future behavior is fully determined by their initial conditions with no random elements involved, the long-term prediction of its chaotic trajectories is still impossible.

Chaotic behavior of nonlinear or piecewise-linear deterministic systems has been extensively studied using mathematical theory since Li-Yorke~\cite{LY}; see, e.g., \cite{Dev}. Dynamical systems that exhibit chaos with certain control actions also are not new to the control community. For example, in \cite{UH85, UH87} it has been demonstrated that quantization can induce chaotic behavior in digital feedback control systems; in nonlinear/piecewise-linear and adaptive control settings, chaotic behavior has also been demonstrated via continuous nonlinear/piecewise-linear feedback control (cf., e.g., \cite{BBW, RAC, LZCY}).
In many situations, identifying the presence of chaos can be a great advantage for feedback control design to stabilize those chaotic trajectories; see, for example, \cite{OGY90, AGOY, BGLMM} and references therein.

There is no doubt about that inclusion/switched systems provide a convenient method for modeling a wide variety of complex dynamical systems. Unfortunately, while the modeling paradigm itself is quite straightforward, the analysis is highly nontrivial; this is because even simple inclusion/switched dynamical systems may exhibit very complex dynamics such as chaotic behavior. In \cite{CSR}, Chase \textit{et al.} presented an example to illustrate how chaotic behavior can arise when switching between low-dimensional linear vector fields by choosing a piecewise-linear expanding map on an interval as the transition function of switching.
In \cite{LTZC}, for a continuous-time switched system that consists of two $3$-dimensional inhomogeneous linear vector fields one of which is of the form
\begin{equation*}
\dot{x}=\left[\begin{matrix}a&b&0\\-b&a&0\\0&0&c\end{matrix}\right]x+\left[\begin{matrix}0\\0\\-d\end{matrix}\right]\quad \textrm{with an expanding equilibrium }x_1^*=\left[\begin{matrix}0\\0\\d/c\end{matrix}\right]
\end{equation*}
and the other of which is of the form
\begin{equation*}
\dot{x}=\left[\begin{matrix}f&0&0\\0&g&h\\0&-h&g\end{matrix}\right]x\quad \textrm{ with a contractive equilibrium }x_2^*=\left[\begin{matrix}0\\0\\0\end{matrix}\right],
\end{equation*}
Liu \textit{et al.} constructed a particular switching rule and gave a numerical simulation to illustrate the chaotic dynamical behavior near the contractive equilibrium $x_2^*$. A similar construction also leads to chaotic behavior near the equilibrium $0$ for continuous-time switched linear system
\begin{equation*}
\dot{x}(t)=A_{\sigma(t)}x(t)+B_{\sigma(t)}u(t),
\end{equation*}
with certain feedback controller $u(t)=F(t,x(t))$ as was shown in \cite{XZZ}.

The above mentioned results are mainly based on the observation of numerical simulations. In current literature, sufficient condition that leads to chaotic dynamics for linear inclusion/switched system as well as theoretical justification remain unsolved. These motivate us to study the following basic question for linear inclusion/switched systems without feedback controls, which is raised by Shorten \textit{et al.}~\cite{SWMWK}:

\begin{OP}[{\cite[Section~1.\,(iv)]{SWMWK}}]
Whether it is possible to determine if a inclusion/switched system can exhibit chaotic behavior for a given set of constituent linear subsystems.
\end{OP}

For inclusion/switched dynamical systems, it is evident that the chaotic behaviors
depend not only on the constituent subdynamics but also on the rule which
orchestrates the switching.

In this paper, we shall employ the idea of Li-Yorke chaos to give a mathematical definition of chaos for System~(\ref{eq1.1}) and then consider when a discrete-time linear
inclusion dynamical system will exhibit the irregular/chaotic dynamical behavior. Our main results---Theorems~\ref{thm1.2} and \ref{thm1.3} below and Theorem~\ref{thm5.3} in Section~\ref{sec5} partially solve the above Open problem of Shorten \textit{et al.}~\cite{SWMWK}.

The obtained results are useful for us to construct examples of discontinuity of Lyapunov exponents (cf.~Remark~\ref{rem1.6} below). This is itself an interesting topic in smooth ergodic theory.

\subsection{Basic mathematical concepts}\label{sec1.1}
Throughout this paper, let $\bK=\{1,\dotsc,K\}$ be endowed with the discrete topology. Let $S_1,\dotsc,S_K$ be $K$ real $d$-by-$d$ matrices, where $K\ge2$ and $d\ge1$ are two integers.
This then induces a discrete-time linear
inclusion dynamical system described by
\begin{equation}\label{eq1.1}
x_n\in\left\{S_kx_{n-1}\right\}_{k\in\bK},\quad x_0\in\mathbb{R}^d\textrm{ and }n\ge1,
\end{equation}
where $x_0$ is called an initial state of this dynamical system. Write
\begin{equation}\label{eq1.2}
\varSigma_{\bK}^+=\{\sigma\colon\mathbb{N}\rightarrow\bK\},\quad \textrm{where }\mathbb{N}=\{1,2,\dotsc\},
\end{equation}
which is as usual equipped with the product topology identifying with $\bK^\mathbb{N}$. Then $\varSigma_{\bK}^+$ is a compact topological space and its topology is compatible with the metric given by
\begin{equation*}
d(\sigma,\sigma^\prime)=\sum_{n=1}^{+\infty}\frac{|\sigma(n)-\sigma^\prime(n)|}{K^{n}}\quad\forall \sigma,\sigma^\prime\in\varSigma_{\bK}^+.
\end{equation*}
Under this product topology, the classical one-sided Markovian shift map
\begin{equation*}
\theta^+\colon\varSigma_{\bK}^+\rightarrow\varSigma_{\bK}^+\quad\textrm{by }\sigma=(\sigma(n))_{n\ge1}\mapsto\theta^+(\sigma)=(\sigma(n+1))_{n\ge1}
\end{equation*}
is a continuous, surjective, and non-injective transformation. 

Let $\|\cdot\|$ be the Euclidean vector norm on $\mathbb{R}^d$ such that $\|x\|=\sqrt{x^\mathrm{T}\cdot x}$ for each $x\in\mathbb{R}^d$.

Then for any $\sigma\in\varSigma_{\bK}^+$ and to any initial state $x_0\in\mathbb{R}^d$, the corresponding state trajectory $(x_n(x_0,\sigma))_{n\ge1}$ of System (\ref{eq1.1}), governed by the switching law $\sigma$, is defined as
\begin{equation*}
x_n(x_0,\sigma)=S_{\sigma(n)}x_{n-1}=S_{\sigma(n)}\dotsm S_{\sigma(1)}x_0,\quad \forall n\ge1.
\end{equation*}
Usually this state trajectory $(x_n(x_0,\sigma))_{n\ge1}$ is said to to be \textit{asymptotically divergent} if and only if
\begin{equation*}
\lim_{n\to\infty}\|x_n(x_0,\sigma)\|=\infty;
\end{equation*}
it is said to be \textit{asymptotically stable} if and only if
\begin{equation*}
\lim_{n\to\infty}\|x_n(x_0,\sigma)\|=0.
\end{equation*}
In addition, the \textit{$\omega$-limit set}, written as $\omega(x_0,\sigma)$, is the set of all limits of the state trajectory $(x_n(x_0,\sigma))_{n\ge1}$ in the Banach space $(\mathbb{R}^d,\|\cdot\|)$. We note here that  $\omega(x_0,\sigma)=\varnothing$ if and only if $(x_n(x_0,\sigma))_{n\ge1}$ is asymptotically divergent; and $\omega(x_0,\sigma)=\{0\}$ if and only if $(x_n(x_0,\sigma))_{n\ge1}$ is asymptotically stable.

Recall that a subset of the topological space $\varSigma_{\bK}^+$ is said to be $G_\delta$ if it is the intersection of a countable collection of open subsets of $\varSigma_{\bK}^+$; and a subset of $\varSigma_{\bK}^+$ is said to be \textit{residual} if it contains a dense $G_\delta$-set. So a residual subset is a large set from the point of view of general topology. If a dynamical property holds based on a residual subset of $\varSigma_{\bK}^+$, then we say this property is \textit{generic} in $\varSigma_{\bK}^+$.

System $(\ref{eq1.1})$ has recently been found in many real applications. For the theoretic and applied importance of the study of System (\ref{eq1.1}), readers may see, e.g., \cite{Lib, SWMWK, SG}. If System (\ref{eq1.1}) is absolutely/uniformly stable, then $x_n(x_0,\sigma)$ converges asymptotically to the equilibrium $0$ as $n\to\infty$ for all switching laws $\sigma\in\varSigma_{\bK}^+$ and all initial states $x_0\in\mathbb{R}^d$. However, the vast majority of dynamical systems defined by (\ref{eq1.1}) do not behave in this way, i.e., their dynamics are unstable. Then the state trajectory $(x_n(x_0,\sigma))_{n\ge1}$ may be much more complex than that $(x_n(x_0,\sigma))_{n\ge1}$ is asymptotically divergent and that $(x_n(x_0,\sigma))_{n\ge1}$ has a nontrivial compact $\omega$-limit set $\omega(x_0,\sigma)$. It is often very unpredictable under the ``generic'' switching laws in $\varSigma_{\bK}^+$.

In some literature, for example \cite{Fel, CHM09, CHM10}, the complexity of System~(\ref{eq1.1}) is described by the existence of a dense trajectory $(x_n(x_0,\sigma))_{n\ge1}$ in $\mathbb{R}^d$ (i.e. $\overline{\{x_n(x_0,\sigma)\,|\,n=1,2,\dotsc\}}=\mathbb{R}^d$). However, for some unstable system, there exists no such a dense state trajectory at all. For example, for $S_1=\frac{1}{2}$ and $S_2=2$, the induced one-dimensional System (\ref{eq1.1}) does not have any dense trajectories in $\mathbb{R}$; but it still exhibits very complex/unpredictable generic characteristic from the randomly switching and long-terms viewpoint; see Theorem~\ref{thm1.2} below.

To appropriately describe the complexity of the dynamics of the state trajectory $(x_n(x_0,\sigma))_{n\ge1}$ of System (\ref{eq1.1}) as time evolves, we now introduce the dynamical concept---fiber chaos, which is motivated by the sensitive dependence on initial conditions in Li-Yorke's definition of chaos \cite{LY} for nonlinear dynamical systems on intervals, capturing System~(\ref{eq1.1})'s excessively irregular behavior.

\begin{defn}\label{def1.1}
A switching law $\sigma\colon \mathbb{N}\rightarrow\bK$ in $\varSigma_{\bK}^+$ is said to be \textit{fiber-chaotic} for System (\ref{eq1.1}) if for \uwave{all initial states $x_0\in\mathbb{R}^d\setminus\{0\}$},
\begin{equation*}
\liminf_{n\to+\infty}\|x_n(x_0,\sigma)\|=0\quad \textrm{and}\quad \limsup_{n\to+\infty}\|x_n(x_0,\sigma)\|=\infty.
\end{equation*}
System (\ref{eq1.1}) is called \textit{generically} (resp. \textit{densely}) fiber-chaotic if its fiber-chaotic switching laws form a residual (resp. dense) subset of the topological space $\varSigma_{\bK}^+$.
\end{defn}

Then so long as $\sigma\in\varSigma_{\bK}^+$ is fiber-chaotic for System~(\ref{eq1.1}), for any $x_0\not=0$, the $\omega$-limit set $\omega(x_0,\sigma)$ must contain $0$ and be unbounded. In fact, $\sigma$ in $\varSigma_{\bK}^+$ is fiber-chaotic for System (\ref{eq1.1}) if and only if $\omega(x_0,\sigma)$ contains $0$ and is unbounded for any nonzero initial state $x_0$ in $\mathbb{R}^d$.

Let $f\colon I\rightarrow I$ be a continuous transformation of the unit interval $I=[0,1]$; then two distinct points $w,z\in I$ are called a Li-Yorke pair~\cite{LY} if
$$
\liminf_{n\to\infty}|f^n(w)-f^n(z)|=0\quad\textrm{and}\quad\limsup_{n\to\infty}|f^n(w)-f^n(z)|>0.
$$
Since $x_n(x_0,\sigma)-x_n(y_0,\sigma)=x_n(x_0-y_0,\sigma)$ for any two initial states $x_0,y_0\in\mathbb{R}^d$, it is easily seen that a switching law $\sigma\in\varSigma_{\bK}^+$ is fiber-chaotic for System $(\ref{eq1.1})$ if and only if
\begin{subequations}
\begin{align}
&\liminf_{n\to+\infty}\|x_n(x_0,\sigma)-x_n(y_0,\sigma)\|=0 \label{eq1.3a}\\
\intertext{and}
&\limsup_{n\to+\infty}\|x_n(x_0,\sigma)-x_n(y_0,\sigma)\|=\infty \label{eq1.3b}
\end{align}\end{subequations}
for all $x_0,y_0\in\mathbb{R}^d$ with $x_0\not=y_0$. So any two distinct initial states $x_0,y_0$ in $\mathbb{R}^d$ are a ``Li-Yorke pair'' of the output
\begin{equation*}
(x_n(\cdot,\sigma))_{n\ge1}\colon\mathbb{R}^d\rightarrow\mathbb{R}^d\times\mathbb{R}^d\times\mathbb{R}^d\times\dotsm
\end{equation*}
and $(x_n(\cdot,\sigma))_{n\ge1}$ is extremely sensitive to initial conditions, if the switching law $\sigma$ is fiber-chaotic for System (\ref{eq1.1}).

In topological dynamical system, $\mathrm{(\ref{eq1.3a})}$ and $\mathrm{(\ref{eq1.3b})}$ are respectively called the proximal and distal properties. However, our distal property $\mathrm{(\ref{eq1.3b})}$ is much more stronger than the general Li-Yorke's one~\cite{LY} that only requires that
\begin{equation*}
\limsup_{n\to+\infty}\|x_n(x_0,\sigma)-x_n(y_0,\sigma)\|>0.
\end{equation*}
This sensitivity to initial conditions means that any two state trajectories of System~(\ref{eq1.1}) governed by the same fiber-chaotic switching law $\sigma$ will be bound to get close together for a while, as time evolves, and then to go far away from each other for a while, and such sharply oscillating dynamics will be repeated infinitely that leads to irregular, complex dynamical behaviors.

A recent notable work for the study of chaotic switching of System~(\ref{eq1.1}) was introduced by Balde and Jouan in~\cite{BJ}. Their study only focused on the switching mechanism and their definition of chaotic switching laws can be found in Definition~\ref{def3.1} in Section~\ref{sec3}. Strictly speaking, under their framework, a chaotic switching law may even not produce a chaotic dynamics of System~(\ref{eq1.1}) in the classical sense, such as in terms of Li-York's definition. This is due to the dynamics of System~(\ref{eq1.1}) not only depends on the switching law but also the structure of the constituent subsystems. Only focusing on switching mechanism is not enough to describe the chaotic dynamics of System~(\ref{eq1.1}). For example, if the joint spectral radius of $\{S_1,..., S_K\}$ is strictly less than one, the state trajectories of System~(\ref{eq1.1}) will always converge to zero regardless if the switching is chaotic or not in the sense of Balde and Jouan in~\cite{BJ}. Therefore it is essential to consider the structure of subsystems as well in order to capture the characteristics of the dynamics of System~(\ref{eq1.1}). One of the motivations of this paper is trying to identify sufficient conditions for which System~(\ref{eq1.1}) appears to be fiber-chaotic.

In Definition~\ref{def1.1}, the word ``fiber'' is used there to distinguish the classical chaos purely respecting the one-sided shift dynamical system $\theta^+\colon\varSigma_{\bK}^+\rightarrow\varSigma_{\bK}^+$ which is of course chaotic in the sense of Li-Yorke. In addition, since for a single linear transformation $L\colon\mathbb{R}^d\rightarrow\mathbb{R}^d$ there is no complex dynamics, we consider only $K\ge2$.
\subsection{Main statements}\label{sec1.2}
Let $f\colon M^n\rightarrow M^n$ be a $\mathrm{C}^{1+\alpha}$-diffeomorphism of a closed $n$-dimensional manifold $M^n$, which is volume-preserving. If $f$ has positive Lyapunov exponents then from Pesin's formula, $f$ has positive entropy. And hence $f$ is chaotic in the sense of Li-Yorke (cf.~\cite{BGKM}); that is, there is an uncountable set $\Delta\subset M^n$ such that any two distinct points $x,y\in\Delta$ are a Li-Yorke pair of $f$, i.e.,
\begin{equation*}
\liminf_{n\to+\infty}\textrm{dist}(f^n(x),f^n(y))=0\quad
\textrm{and}\quad
\limsup_{n\to+\infty}\textrm{dist}(f^n(x),f^n(y))>0,
\end{equation*}
where $\textrm{dist}(\cdot,\cdot)$ is a compatible metric on $M^n$.
However, this is not the case for linear inclusion dynamical system. For example, let
$$
S_1=\left(\begin{matrix}2&0\\0&\frac{1}{2}\end{matrix}\right)\quad \textrm{and}\quad S_2=\left(\begin{matrix}3&0\\0&\frac{1}{3}\end{matrix}\right).
$$
Then System (\ref{eq1.1}) induced by $S_1$ and $S_2$ is area-preserving because of $\det S_1=\det S_2=1$ and it has positive Lyapunov exponents at almost every switching laws $\sigma$ (cf.~Section~\ref{sec3.2} for the definition of Lyapunov exponents); but obviously there is no fiber chaos for System (\ref{eq1.1}).

For an Anosov automorphism $f$ of the $n$-dimensional torus $\mathbb{T}^n$, it is Li-Yorke chaotic because of the existence of Smale horseshoe. However, its derivative $Df\colon\mathbb{T}^n\times\mathbb{R}^n\rightarrow\mathbb{T}^n\times\mathbb{R}^n$ is not necessarily fiber-chaotic in the sense of Definition~\ref{def1.1}.

In this paper we shall present, for System $\mathrm{(\ref{eq1.1})}$, two simple mechanisms of generating the fiber-chaotic dynamics described as in Definition~\ref{def1.1}, as follows:

\begin{theorem}\label{thm1.2}
Let $S_1,\dotsc,S_K$ all be nonsingular $d$-by-$d$ matrices. Then System $(\ref{eq1.1})$ is generically fiber-chaotic, if there are two finite-length words
$(i_1,\dotsc,i_m)\in\bK^m$ and $(j_1,\dotsc,j_n)\in\bK^n$
such that
\begin{equation*}
\|S_{i_m}\dotsm S_{i_1}\|<1<\|S_{j_n}\dotsm S_{j_1}\|_{\mathrm{co}}.
\end{equation*}
\end{theorem}

Here $\|\cdot\|$ and $\|\cdot\|_{\mathrm{co}}$ denote the usual matrix maximum norm and minimum/co- norm, respectively, defined by
\begin{equation*}
\|A\|=\max_{x\in\mathbb{R}^d,\|x\|=1}\|Ax\|\quad \textrm{and}\quad\|A\|_{\textrm{co}}=\min_{x\in\mathbb{R}^d,\|x\|=1}\|Ax\|
\end{equation*}
for any $d$-by-$d$ matrix $A$.

In order to prove Theorem~\ref{thm1.2}, we shall employ the analogous idea of constructing switching law $\sigma$ as in~\cite{LTZC}; that is, for any initial state $x_0\,(\not=0)\in\mathbb{R}^d$, when $S_{\sigma(k)}\dotsm S_{\sigma(1)}x_0$ is far away from the origin $0$, we will activate $S_{i_m}\dotsm S_{i_1}$ until
$(S_{i_m}\dotsm S_{i_1})^\ell S_{\sigma(k)}\dotsm S_{\sigma(1)}x_0$ closes sufficiently to $0$ for some $\ell$; when $S_{\sigma(k^\prime)}\dotsm S_{\sigma(1)}x_0$ is close to $0$, we will activate $S_{j_n}\dotsm S_{j_1}$ until
$(S_{j_n}\dotsm S_{j_1})^{\ell^\prime} S_{\sigma(k^\prime)}\dotsm S_{\sigma(1)}x_0$ is far away from $0$ for some $\ell^\prime$; and alternately do this. This idea has been proven by using numerical result in \cite{LTZC}. Our mathematical proof will be presented in Section~\ref{sec4.1} using topology. But different with \cite{LTZC}, our procedure is uniform for the initial state $x_0$; since the fiber-chaoticity of $\sigma$ is independent of $x_0$ by Definition~\ref{def1.1}. See Section~\ref{sec4.1} for the details.

The following is another mechanism of generating fiber chaos for the case where System~(\ref{eq1.1}) has neither a asymptotically stable subsystem nor an asymptotically divergent subsystem, which is completely beyond the situations of \cite{CSR, LTZC, XZZ}.

\begin{theorem}\label{thm1.3}
Let $S_1,\dotsc,S_K$ all be real, $2$-by-$2$, and nonsingular matrices. If there exists an asymptotically stable state trajectory $(x_n(\bx_0,\sigma_0))_{n\ge1}$, an asymptotically divergent state trajectory $(x_n(\by_0,\sigma_1))_{n\ge1}$, and an irrational rotation subsystem, say
\begin{equation*}
S_1=\left(\begin{matrix}\cos2\pi\alpha&\sin2\pi\alpha\\-\sin2\pi\alpha&\cos2\pi\alpha\end{matrix}\right)\quad \textrm{where }\alpha\in(0,1)\textrm{ is irrational},
\end{equation*}
then System~$(\ref{eq1.1})$ is densely fiber-chaotic.
\end{theorem}

Note that if $\alpha$ is rational, then the statement of Theorem~\ref{thm1.3} is not necessarily true; see the counterexample presented in Section~\ref{sec5.2}.

We shall prove that the set of all fiber-chaotic switching laws of System $(\ref{eq1.1})$ is of measure $0$ for any ergodic measure of $\theta^+$ that has a nonzero Lyapunov exponent (cf.~Proposition~\ref{prop3.6} below). So this generic/dense fiber-chaotic phenomenon is ignorable from the viewpoint of ergodic theory. However, under the situations of Theorems~\ref{thm1.2}, \ref{thm1.3} and~\ref{thm4.7} it is not ignorable from the viewpoint of topology.

Although Definition~\ref{def1.1} does not involve directly any Lyapunov exponents of System $(\ref{eq1.1})$, yet the fiber-chaos is essentially related to this fundamental index. We shall show that a necessary condition of fiber-chaos is the coexistence of at least one positive Lyapunov exponent, one negative Lyapunov exponent and one zero Lyapunov exponent for System~(\ref{eq1.1}). This shows that the fiber chaos really captures the complex dynamical behavior of System~(\ref{eq1.1}); see Theorem~\ref{thm5.2} stated in Section~\ref{sec5}.

On the other hand, we shall study the fiber chaos of periodically stable linear inclusion systems. Our Theorem~\ref{thm5.3} below shows that this kind of system is ``simple'' from the viewpoint of fiber chaos.

\begin{remark}
Let $\varphi\colon\varSigma_{\bK}^+\rightarrow\mathbb{R}$ be H\"{o}lder continuous such that $\varphi$ is not cohomologous to $0$ (i.e., one cannot find a $\psi\in C(\varSigma_{\bK}^+,\mathbb{R})$ and $c\in\mathbb{R}$ so that $\varphi=\psi-\psi\circ\theta^++c$). Then Barreira and Schmeling proved that the following statements hold (\cite[Theorem~2.1]{BS}):
\begin{enumerate}
\item[(1)] The irregular set of the Birkhoff average of $\varphi$, defined by
\begin{equation*}
\mathfrak{B}(\varphi)=\left\{\sigma\in\varSigma_{\bK}^+\colon\lim_{n\to\infty}\frac{1}{n}\sum_{i=0}^{n-1}\varphi(\sigma(\cdot+i))\textrm{ does not exist}\right\},
\end{equation*}
is a proper dense subset of $\varSigma_{\bK}^+$.

\item[(2)] $\theta^+$ has the full topological entropy restricted to $\mathfrak{B}(\varphi)$.
\end{enumerate}
By the Birkhoff ergodic theorem, $\mathfrak{B}(\varphi)$ has total measure $0$ in the sense of $\theta^+$-invariant measures (cf. Section~\ref{sec3.2} for the precise definition); if $\varphi$ is cohomologous to $0$, then $\mathfrak{B}(\varphi)=\varnothing$.

Applying the above result of Barreira and Schmeling to System~(\ref{eq1.1}), if $S_1,\dotsc,S_K$ are conformal (i.e., each $S_k$ is a multiple of an isometry) and let $\bvarphi(\sigma)=\log\|S_{\sigma(1)}\|$, then
$$\|S_{i_1}\dotsm S_{i_n}\|_{\mathrm{co}}=\|S_{i_1}\dotsm S_{i_n}\|=\|S_{i_1}\|\dotsm \|S_{i_n}\|$$
and hence for every $\sigma\in\mathfrak{B}(\bvarphi)$ and for all $x_0\in\mathbb{R}^d\setminus\{0\}$,
$$
\liminf_{n\to\infty}\frac{1}{n}\log\|x_n(x_0,\sigma)\|\not=\limsup_{n\to\infty}\frac{1}{n}\log\|x_n(x_0,\sigma)\|.
$$
This shows that the state trajectory $(x_n(x_0,\sigma))_{n\ge1}$ is infinitely oscillated, but we cannot, only for this reason, claim that such a $\sigma$ is fiber-chaotic for System~(\ref{eq1.1}). In addition, generally the condition
\begin{equation*}
\liminf_{n\to\infty}\|x_n(x_0,\sigma)\|\not=\limsup_{n\to\infty}\|x_n(x_0,\sigma)\|
\end{equation*}
does not imply that
\begin{equation*}
\liminf_{n\to\infty}\frac{1}{n}\log\|x_n(x_0,\sigma)\|\not=\limsup_{n\to\infty}\frac{1}{n}\log\|x_n(x_0,\sigma)\|.
\end{equation*}
Moreover, $\|S_{i_1}\dotsm S_{i_n}\|\not=\|S_{i_1}\|\dotsm \|S_{i_n}\|$ in general for $d\ge2$.

So our fiber-chaotic switching law is completely different with the irregular point of the Birkhoff average of a continuous function.
\end{remark}

\begin{remark}
Under the hypotheses of Theorem~\ref{thm1.2}, one can find some $\varepsilon>0$ such that for any $S_1^\prime, \dotsc, S_K^\prime\in\mathbb{R}^{d\times d}$, if
\begin{equation*}
\|S_1-S_1^\prime\|<\varepsilon,\ \dotsc,\ \|S_K-S_K^\prime\|<\varepsilon,
\end{equation*}
then the linear inclusion system induced by $\{S_1^\prime, \dotsc, S_K^\prime\}$ is also generically fiber-chaotic.
\end{remark}

\begin{remark}\label{rem1.6}
It is a widely known fact that Lyapunov exponents are very sensitive to the base points, but it is rather difficult to find examples of this kind. Theorems~\ref{thm1.2} and \ref{thm1.3} is useful for us to construct such examples. For example, let
\begin{equation*}
S_1=\left[\begin{matrix}\cos2\pi\alpha&\sin2\pi\alpha\\-\sin2\pi\alpha&\cos2\pi\alpha\end{matrix}\right], S_2=\left[\begin{matrix}1&1\\0&1\end{matrix}\right]\textrm{ and } S_3=\left[\begin{matrix}1&0\\1&\frac{1}{2}\end{matrix}\right],
\end{equation*}
where $\alpha\in(0,1)$ is irrational. Clearly, $\{S_1,S_2,S_3\}$ satisfies the conditions of Theorem~\ref{thm1.3}. Then from Theorem~\ref{thm1.3} and Theorem~\ref{thm5.2}, it follows that there exist some $\bsigma\in\varSigma_{\bK}^+$ and $\bx_0\in\mathbb{R}^2\setminus\{0\}$ such that the Lyapunov exponent of System~(\ref{eq1.1}) at the data $(\bx_0,\bsigma)$
\begin{equation*}
\lambda(\bx_0,\bsigma):=\lim_{n\to\infty}\frac{1}{n}\log\|x_n(\bx_0,\bsigma)\|> (\textrm{resp.~}<)\ 0.
\end{equation*}
We can claim that $\lambda(\bx_0,\cdot)$ is not continuous at the base point $\bsigma$. In fact Theorem~\ref{thm1.3} follows that for any $\epsilon>0$, there is at least one $\sigma\in\varSigma_{\bK}^+$ with $d(\bsigma,\sigma)<\epsilon$ such that either
\begin{equation*}
\lambda(\bx_0,\sigma)=\lim_{n\to\infty}\frac{1}{n}\log\|x_n(\bx_0,\bsigma)\|=0
\end{equation*}
or the above limit does not exist at all. This proves the discontinuity of the exponent function $\lambda(\bx_0,\cdot)$ at the point $\bsigma$.
\end{remark}
\subsection{Outline}
This paper is organized as follows: In Section~\ref{sec2}, we will introduce a dual concept of the joint/generalized spectral radius and prove an Elsner-type reduction theorem for System~(\ref{eq1.1}). In Section~\ref{sec3} we shall study the topological structure of a nonchaotic switching law and recall the definition of Lyapunov exponents.
We will prove our main results Theorems~\ref{thm1.2} and \ref{thm1.3} in Section~\ref{sec4}.
In Section~\ref{sec5}, we will prove the coexistence of positive, zero and negative Lyapunov exponents of a fiber-chaotic system (cf.~Theorem~\ref{thm5.2}) and show that every periodically stable inclusion system does not have any fiber-chaotic behaviors (cf.~Theorem~\ref{thm5.3}). So, a periodically stable inclusion system is ``simple'' from our viewpoint of fiber chaos. Finally we will end this paper with some questions related to the fiber-chaos of System (\ref{eq1.1}) in Section~\ref{sec6} for us to further study.

\section{Joint spectral co-radius and a reduction theorem}\label{sec2}
As was pointed out in the Introduction, the fiber chaos implies that there exists at least one negative Lyapunov exponent for System~(\ref{eq1.1}) (cf.~Theorem~\ref{thm5.2}). To prove this, we need to introduce a concept---the joint spectral co-radius---and then prove an Elsner-type reduction theorem for System~(\ref{eq1.1}).

Let System~(\ref{eq1.1}) be based on the $K$ real $d$-by-$d$ nonsingular matrices $S_1,\dotsc,S_K$ throughout the sequel of this section.

\subsection{The joint spectral co-radius}\label{sec2.1}
According to Rota and Strang~\cite{RS}, the nonnegative real number
\begin{equation*}
\widehat{\pmb{\rho}}=\lim_{n\to\infty}\max_{\sigma\in\varSigma_{\bK}^+}\sqrt[n]{\|S_{\sigma(n)}\dotsm S_{\sigma(1)}\|} \quad \left(=\inf_{n\ge1}\max_{\sigma\in\varSigma_{\bK}^+}\sqrt[n]{\|S_{\sigma(n)}\dotsm S_{\sigma(1)}\|}\right),
\end{equation*}
is called the \textit{joint spectral radius} of System~(\ref{eq1.1}).

By $\rho(A)$ we denote the usual spectral radius of a square matrix $A$; that is, if $\lambda_1,\dotsc,\lambda_r$ are its all distinct eigenvalues, then $\rho(A)=\max_{1\le i\le r}|\lambda_i|$. According to Daubechies and Lagarias~\cite{DL} the \textit{generalized spectral radius} of System~(\ref{eq1.1}) is defined by
\begin{equation*}
\pmb{\rho}=\limsup_{n\to\infty}\max_{\sigma\in\varSigma_{\bK}^+}\sqrt[n]{\rho(S_{\sigma(n)}\dotsm S_{\sigma(1)})}.
\end{equation*}

Then there holds the following important relationship:

\begin{BW}[{\cite{BW}; also see \cite{El,SWP,CZ,Boc,Dai-JMAA,Dai-LAA}}]
For System~$(\ref{eq1.1})$, there holds $\widehat{\pmb{\rho}}=\pmb{\rho}$.
\end{BW}

In light of this formula, we can identify $\widehat{\pmb{\rho}}$ with $\pmb{\rho}$ from now on.

Let $\mathbb{R}^{d\times d}$ be the space of all $d\times d$ real matrices. For any nonsingular $A\in\mathbb{R}^{d\times d}$, write $\rho_{\mathrm{co}}^{}(A)=1/\rho(A^{-1})$.
Dually we now introduce the following concepts:

\begin{defn}\label{def2.1}
The nonnegative real number
\begin{equation*}
\widehat{\pmb{\rho}}_{\mathrm{co}}^{}=\lim_{n\to\infty}\min_{\sigma\in\varSigma_{\bK}^+}\sqrt[n]{\|S_{\sigma(n)}\dotsm S_{\sigma(1)}\|_{\mathrm{co}}}\quad \left(=\sup_{n\ge1}\min_{\sigma\in\varSigma_{\bK}^+}\sqrt[n]{\|S_{\sigma(n)}\dotsm S_{\sigma(1)}\|_{\mathrm{co}}}\right),
\end{equation*}
is called the \textit{joint spectral co-radius} of System~(\ref{eq1.1}). The nonnegative real number
\begin{equation*}
\pmb{\rho}_{\mathrm{co}}^{}=\liminf_{n\to\infty}\min_{\sigma\in\varSigma_{\bK}^+}\sqrt[n]{\rho_{\mathrm{co}}^{}(S_{\sigma(n)}\dotsm S_{\sigma(1)})}
\end{equation*}
is called the \textit{generalized spectral co-radius} of System~(\ref{eq1.1}).
\end{defn}

Since System~(\ref{eq1.1}) is nonsingular and $\|AB\|_{\mathrm{co}}\ge\|A\|_{\mathrm{co}}\cdot\|B\|_{\mathrm{co}}$ for any two $d\times d$ matrices, the above joint spectral co-radius $\widehat{\pmb{\rho}}_{\mathrm{co}}^{}$ is well defined for System~(\ref{eq1.1}).

From Berger-Wang's formula, it follows easily the following equality:

\begin{lem}\label{lem2.2}
For System~$(\ref{eq1.1})$, there holds that $\widehat{\pmb{\rho}}_{\mathrm{co}}^{}=\pmb{\rho}_{\mathrm{co}}^{}$, and that $\widehat{\pmb{\rho}}_{\mathrm{co}}^{-1}$ and $\pmb{\rho}_{\mathrm{co}}^{-1}$ are the joint and generalized spectral radii of $\{S_1^{-1},\dotsc,S_K^{-1}\}$ respectively.
\end{lem}

\begin{proof}
The statement follows from considering the linear inclusion system induced by $\left\{S_1^{-1}, \dotsc, S_K^{-1}\right\}$ and we thus omit the details here.
\end{proof}

On the other hand, from Lemma~\ref{lem2.2} and the simple fact that for any block-upper-triangular nonsingular matrix $S=\left[\begin{matrix}A&*\\0&B\end{matrix}\right]$ its inverse is
$S^{-1}=\left[\begin{matrix}A^{-1}&\star\\0&B^{-1}\end{matrix}\right]$, we could obtain the following lemma:

\begin{lem}\label{lem2.3}
Let $A_1,\dotsc,A_K\in\mathbb{R}^{d_1\times d_1}$ and $B_1,\dotsc,B_K\in\mathbb{R}^{d_2\times d_2}$ all be nonsingular with $n=d_1+d_2$. If System~$(\ref{eq1.1})$ is based on the following
\begin{equation*}
S_1=\left[\begin{matrix}A_1&*_1\\0&B_1\end{matrix}\right],\ \dotsc,\ S_K=\left[\begin{matrix}A_K&*_K\\0&B_K\end{matrix}\right],\quad \textrm{where } *_1,\dotsc,*_K\in\mathbb{R}^{d_1\times d_2},
\end{equation*}
then the joint spectral co-radius $\widehat{\pmb{\rho}}_{\mathrm{co}}^{}$ of System~$(\ref{eq1.1})$ is just the minimum of the joint spectral co-radii of the systems based, respectively, on $\{A_1,\dotsc,A_K\}$ and $\{B_1,\dotsc,B_K\}$.
\end{lem}

\subsection{An Elsner-type reduction lemma}\label{sec2.2}
Recall that System~$(\ref{eq1.1})$ is called \textit{product bounded} if and only if there exists a constant $1\le\beta<+\infty$ such that
\begin{equation*}
\|S_{\sigma(n)}\dotsm S_{\sigma(1)}\|\le\beta\quad\forall\sigma\in\varSigma_{\bK}^+\textrm{ and }n\ge1.
\end{equation*}
This means that the semigroup $\bigsqcup_{n\in\mathbb{N}}\{S_{\sigma(n)}\dotsm S_{\sigma(1)}\colon\sigma\in\varSigma_{\bK}^+\}$ is a bounded subset of the $d\times d$ matrix space $\mathbb{R}^{d\times d}$.

Then Elsner's reduction theorem may be stated as follows:

\begin{E}[{\cite{El}; see also~\cite{Dai-JMAA}}]
If System~$(\ref{eq1.1})$ is product unbounded such that $\widehat{\pmb{\rho}}=1$, then it is reducible, i.e., each $S_k$ has a common, nontrivial, proper, and invariant linear subspace in $\mathbb{R}^d$.
\end{E}

We should note here that the original Elsner's reduction theorem in~\cite{El} is for complex matrices. However, from the another proof presented in \cite{Dai-JMAA} we can see that the statement also holds for real matrices.

Dual to product boundedness, we introduce the following condition:

\begin{defn}\label{def2.4}
System~(\ref{eq1.1}) is called \textit{product co-bounded} if there exists a constant $0<\alpha\le 1$ such that
\begin{equation*}
\|S_{\sigma(n)}\dotsm S_{\sigma(1)}\|_{\mathrm{co}}\ge\alpha\quad\forall\sigma\in\varSigma_{\bK}^+\textrm{ and }n\ge1.
\end{equation*}
\end{defn}

Then we will need the following dual form of Elsner's reduction theorem later:

\begin{lem}\label{lem2.5}
If System~$(\ref{eq1.1})$ is product co-unbounded such that $\widehat{\pmb{\rho}}_{\mathrm{co}}^{}=1$, then it is reducible in the sense of Elsner.
\end{lem}

\begin{proof}
It is clear that $\{S_1,\dotsc,S_K\}$ is reducible if and only if $\{S_1^{-1},\dotsc,S_K^{-1}\}$ is reducible. So we only need prove that $\{S_1^{-1},\dotsc,S_K^{-1}\}$ is reducible. For this, since $\|A\|_{\mathrm{co}}=1/\|A^{-1}\|$ and System~(\ref{eq1.1}) is product co-unbounded with $\widehat{\pmb{\rho}}_{\mathrm{co}}^{}=1$, we can obtain that the inclusion system based on $\{S_1^{-1},\dotsc,S_K^{-1}\}$ is product unbounded with the joint spectral radius $1$. Now the statement comes from Elsner's reduction theorem. This completes the proof of Lemma~\ref{lem2.5}.
\end{proof}
\section{Chaotic switching laws}\label{sec3}

This section will be devoted to comparing our definition of fiber-chaotic switching law with the chaos of Balde and Jouan introduced in \cite{BJ}. In addition, we shall introduce the Lyapunov exponents and study the fiber-chaotic dynamics from a viewpoint of ergodic theory.
Let $\{S_1,\dotsc,S_K\}\subset\mathbb{R}^{d\times d}$, not necessarily nonsingular, and then we still consider the induced linear inclusion system
\begin{equation*}
x_n\in\left\{S_1,\dotsc,S_K\right\}x_{n-1},\quad x_0\in\mathbb{R}^d\textrm{ and }n\ge1.\eqno{(\ref{eq1.1})}
\end{equation*}

\subsection{Discrete-time version of chaotic input of Balde and Jouan}\label{sec3.1}
The discrete-time version of a recent definition of chaotic input of Balde and Jouan \cite[Definition~1, p.~1055]{BJ} may be described as follows, which is completely independent of the matrices $S_1,\dotsc,S_K$.

\begin{defn}\label{def3.1}
A switching law $\sigma=(\sigma(n))_{n\ge1}\in\varSigma_{\bK}^+$ is called {\it nonchaotic} in the sense of Balde and Jouan, if to any positive integer sequence $\langle n_i\rangle_{i\ge1}$ with $n_i\nearrow+\infty$ and for any $m\ge1$, there corresponds some $\delta$ with $2\le\delta\le m+1$ such that for all $\ell_0\ge1$, there exists an $\ell\ge \ell_0$ so that $\sigma$ is constant restricted to some subinterval of $[n_{\ell}, n_{\ell}+m]$ of length greater than or equal to $\delta$. If $\sigma$ is not nonchaotic, then we say it is \textit{chaotic} in the sense of Balde and Jouan.
\end{defn}

We note that for any $\varSigma_{\bK}^+$ where $K\ge2$, there always exist Balde-Jouan nonchaotic switching laws, such as the law $\pmb{\sigma}$ given in Lemma~\ref{lem3.4} below is a nontrivial nonchaotic switching law in the sense of Definition~\ref{def3.1}.

Clearly, an eventually constant switching law $\sigma$ (i.e. $\sigma(n)\equiv k$ for all $n$ sufficiently large, for some $1\le k\le K$), is nonchaotic in the sense of Balde and Jouan; meanwhile, it is not fiber-chaotic in the sense of our Definition~\ref{def1.1} too. In fact, we can obtain a more general result.

\begin{prop}\label{prop3.2}
If $\sigma\in\varSigma_{\bK}^+$ is an eventually periodical switching law, then it is not fiber-chaotic for System $(\ref{eq1.1})$ in the sense of Definition~\ref{def1.1}.
\end{prop}

\begin{proof}
Since $\sigma$ is eventually periodical, it can be written as
\begin{equation*}
\sigma=(\sigma(1),\dotsc,\sigma(n_0),\uwave{k_1,\dotsc,k_\pi},\uwave{k_1,\dotsc,k_\pi},\dotsc)
\end{equation*}
for some integer $n_0\ge0$ and some word $(k_1,\dotsc,k_\pi)\in\bK^\pi$ of length $\pi\ge1$.
If $S_{\sigma(n_0)}\dotsm S_{\sigma(1)}$ is singular, then $\sigma$ is obviously not fiber-chaotic. Thus we assume $S_{\sigma(n_0)}\dotsm S_{\sigma(1)}$ is nonsingular from now on.

Simply set $A=S_{k_\pi}\dotsm S_{k_1}$. If the spectral radius $\rho(A)$ of $A$ is less than $1$, then from the classical Gel'fand spectral-radius formula
\begin{equation*}\begin{split}
\lim_{n\to+\infty}\frac{1}{n}\log\|S_{\sigma(n)}\dotsm S_{\sigma(1)}\|&=\lim_{n\to+\infty}\frac{1}{n}\log\|S_{\sigma(n+n_0)}\dotsm S_{\sigma(1+n_0)}\|\\
&=\frac{1}{\pi}\log\rho(A)\\
&<0.
\end{split}\end{equation*}
So,
\begin{equation*}
\lim_{n\to+\infty}\|x_n(x_0,\sigma)\|=0\quad\forall x_0\in\mathbb{R}^d\setminus\{0\},
\end{equation*}
which means that $\sigma$ is not fiber-chaotic for System (\ref{eq1.1}) in the sense of Definition~\ref{def1.1}. If $\rho(A)\ge1$, then one can find a unit vector $\bx_0\in\mathbb{R}^d$ and an eigenvalue $\lambda$ of $A$ with $|\lambda|\ge1$ such that
$$
A^n\bx_0=\lambda^n\bx_0\quad\forall n\ge1,
$$
which implies that for $u=(S_{\sigma(n_0)}\dotsm S_{\sigma(1)})^{-1}\bx_0$,
$$
\liminf_{n\to+\infty}\|x_n(u,\sigma)\|>0,
$$
and so $\sigma$ is not fiber-chaotic for System (\ref{eq1.1}) in the sense of Definition~\ref{def1.1}.

This concludes the statement of Proposition~\ref{prop3.2}.
\end{proof}

However, a nontrivial eventually periodical switching law is always chaotic of Balde and Jouan from Lemma~\ref{lem3.3} below; for example, $\sigma=(\sigma(1)\dotsm\sigma(n_0)121212\dotsm)$.

Balde and Jouan's definition~\ref{def3.1} of chaos only depends on the single switching law $\sigma$ and ignores the structure of System $(\ref{eq1.1})$, which is not enough to capture the essential of chaos of System $(\ref{eq1.1})$. The following lemma gives the key property of a Balde-Jouan nonchaotic switching law.

\begin{lem}\label{lem3.3}
Let $\sigma\in\varSigma_{\bK}^+$ be a nonchaotic switching law in the sense of Balde and Jouan. Then, there exists some symbol $k\in\{1,\dotsc,K\}$ such that for any $\ell\ge1$ and any $\ell^\prime\ge1$, there exists an $n_\ell\ge\ell^\prime$ which satisfies $\sigma(n_\ell+1)=\dotsm=\sigma(n_\ell+\ell)=k$.
\end{lem}

\begin{proof}
First, for the nonchaotic switching law $\sigma$ we can choose a sequence $\langle n_i\rangle_{i\ge1}\nearrow+\infty$ and some $k\in\{1,\dotsc,K\}$, which are such that $n_{i+1}-n_i\nearrow+\infty$ and $\sigma(n_i)=k$ for all $i\ge1$. Now from Definition~\ref{def3.1} with $m=1$, it follows that we can choose a subsequence of $\langle n_i\rangle_{i\ge1}$, still write, without loss of generality, as $\langle n_i\rangle_{i\ge1}$, such that $\sigma(n_i)=\sigma(n_i+1)=k$ for all $i\ge1$. Repeating this procedure for $\langle n_{i}+1\rangle_{i\ge1}$, we can prove the statement of Lemma~\ref{lem3.3}.
\end{proof}

However, our fiber-chaotic property is a kind of dynamical behavior, which discovers the complexity of the structure of the trajectories of System~(\ref{eq1.1}), as to be shown by Lemma~\ref{lem3.4} below. And from Proposition~\ref{prop3.2}, it also depends on the topological structure of the switching law $\sigma$ itself.

\begin{lem}\label{lem3.4}
Let System $(\ref{eq1.1})$ be a $1$-dimensional system defined by
$S_1=\frac{1}{2}$ and $S_2=2$.
Then for System $(\ref{eq1.1})$, the switching law $\pmb{\sigma}\in\varSigma_{\bK}^+$ given as
\begin{equation*}
\pmb{\sigma}=(11,2222,\stackrel{2^{3}\textrm{-folds}}{\overbrace{1\dotsm1}},\stackrel{2^{4}\textrm{-folds}}{\overbrace{2\dotsm2}},\dotsc,\stackrel{2^{2n-1}\textrm{-folds}}{\overbrace{1\dotsm 1}},\stackrel{2^{2n}\textrm{-folds}}{\overbrace{2\dotsm 2}},\dotsc)
\end{equation*}
is fiber-chaotic under the sense of Definition~\ref{def1.1}, but $\pmb{\sigma}$ is nonchaotic in the sense of Balde and Jouan.
\end{lem}

\begin{proof}
The statement comes easily from Definitions~\ref{def1.1} and \ref{def3.1} and we thus omit the details here.
\end{proof}

In fact, we can show this system is generically fiber-chaotic under the sense of Definition~\ref{def1.1} from Theorem~\ref{thm1.2}.
Although the chaos of Balde and Jouan is independent of the matrices $S_1,\dotsc,S_K$, yet if they share a common Lyapunov matrix $P$ such that $S_k^TPS_k-P\le0$ for $1\le k\le K$ (which implies that $\|S_k\|_P\le1$ for all $1\le k\le K$), then it is useful for us to capture the asymptotic stability of System~(\ref{eq1.1}) governed by Balde-Jouan nonchaotic switching laws $\sigma$; for example, see \cite{BJ} and also Theorem~\ref{thm5.5} below.

\subsection{Lyapunov exponents and an ergodic-theoretic viewpoint}\label{sec3.2}
Next using the Lyapunov exponents, we will study a case where the fiber-chaotic behavior does not occur from the ergodic-theoretic viewpoint. For that, let
\begin{equation*}
\theta^+\colon\varSigma_{\bK}^+\rightarrow\varSigma_{\bK}^+;\quad \sigma=(\sigma(n))_{n\ge1}\mapsto\theta^+(\sigma)=(\sigma(n+1))_{n\ge1}
\end{equation*}
be the classical one-sided shift transformation on the compact metrizable space $\varSigma_{\bK}^+$ of all the possible switching laws of System (\ref{eq1.1}), as before.
Then $x_n(x_0,\sigma)=S_{\sigma(n)}\dotsm S_{\sigma(1)}x_0$ for any data $(x_0,\sigma)\in\mathbb{R}^d\times\varSigma_{\bK}^+$.

For reader's convenience, we restate the notion of Lyapunov exponent following the standard way:

\begin{defn}\label{def3.5}
For any switching law $\sigma\in\varSigma_{\bK}^+$ and any nonzero initial state $x_0\in\mathbb{R}^d$, the number
\begin{equation*}
\lambda(x_0,\sigma):=\lim_{n\to\infty}\frac{1}{n}\log\|x_n(x_0,\sigma)\|\in[-\infty, \infty)\quad (\textrm{if the limit exists})
\end{equation*}
is called a \textit{Lyapunov exponent} of System~(\ref{eq1.1}) at $\sigma$ or the \textit{Lyapunov exponent} of System~(\ref{eq1.1}) at the data $(x_0,\sigma)$. If the above limit does not exist for some nonzero initial state $x_0$, then we say that $\sigma$ is a Lyapunov \textit{irregular} point of System~(\ref{eq1.1}).
\end{defn}

Let $\{\theta^+\in B\}=\{\sigma\in\varSigma_{\bK}^+\colon\theta^+(\sigma)\in B\}$. A probability measure $\mu$ on the Borel measurable space $(\varSigma_{\bK}^+,\mathscr{B}(\varSigma_{\bK}^+))$ is said to be \textit{invariant} with respect to $\theta^+$ if and only if $\mu(B)=\mu(\{\theta^+\in B\})$ for all $B\in\mathscr{B}(\varSigma_{\bK}^+)$; further an invariant probability measure $\mu$ is called \textit{ergodic} with respect to $\theta^+$ if either $\mu(B)=0$ or $\mu(B)=1$ whenever $\mu(B\vartriangle\{\theta^+\in B\})=0$, where $\vartriangle$ stands for the symmetric difference of two subsets. Cf., e.g., \cite{Wal}

According to the celebrated Oselede\v{c} multiplicative ergodic theorem~\cite{Ose}, for any $\theta^+$-ergodic measure $\mu$ we can obtain this result: There are $r$ constants (which are called the Lyapunov exponents of System (\ref{eq1.1}) at $\mu$)
$$-\infty\le\lambda_1<\dotsm<\lambda_r<+\infty,\quad \textrm{where }1\le r\le d,$$
and invariant linear subspaces (filtration) of $\mathbb{R}^d$
$$
\{0\}=V^{(0)}(\sigma)\subset V^{(1)}(\sigma)\subset\dotsm\subset V^{(r)}(\sigma)=\mathbb{R}^d\quad (\mu\textrm{-a.e. }\sigma\in\varSigma_{\bK}^+),
$$
such that for $\mu\textrm{-a.e. }\sigma\in\varSigma_{\bK}^+$,
$$
\lambda_i=\lim_{n\to\infty}\frac{1}{n}\log\|x_n(x_0,\sigma)\|\quad \forall x_0\in V^{(i)}(\sigma)\setminus V^{(i-1)}(\sigma),
$$
for all $1\le i\le r$. Here $\lambda_1$ and $\lambda_r$ are called the \textit{minimal} and \textit{maximal} Lyapunov exponents of System (\ref{eq1.1}) at $\mu$, respectively.

\begin{prop}\label{prop3.6}
Let $\mu$ be an ergodic probability measure of the one-sided shift $\theta^+$ on $\varSigma_{\bK}^+$. If System $(\ref{eq1.1})$ has a nonzero Lyapunov exponent at $\mu$, then for $\mu$-a.e. $\sigma\in\varSigma_{\bK}^+$ it is not fiber-chaotic for System $(\ref{eq1.1})$ in the sense of Definition~\ref{def1.1}.
\end{prop}

\begin{proof}
Let $\lambda<0$ be a Lyapunov exponent of System (\ref{eq1.1}) at the ergodic measure $\mu$. Then from the Oselede\v{c} multiplicative ergodic theorem~\cite{Ose}, it follows that for $\mu$-a.e. $\sigma\in\varSigma_{\bK}^+$ there exists a corresponding unit vector, say $u=u(\sigma)\in\mathbb{R}^d$, such that
$$
\lim_{n\to+\infty}\frac{1}{n}\log\|S_{\sigma(n)}\dotsm S_{\sigma(1)}u\|=\lambda.
$$
So $\limsup_{n\to+\infty}\|x_n(u,\sigma)\|=0$.
This shows that for $\mu$-almost every $\sigma\in\varSigma_{\bK}^+$, it is not fiber-chaotic for System (\ref{eq1.1}) in the sense of Definition~\ref{def1.1} because of the lack of the distal property $\mathrm{(\ref{eq1.3b})}$.

Similarly, if System (\ref{eq1.1}) has a Lyapunov exponent $\lambda>0$ at $\mu$ then
$\liminf_{n\to+\infty}\|x_n(u(\sigma),\sigma)\|=+\infty$ for $\mu$-a.e. $\sigma\in\varSigma_{\bK}^+$.
So, there is no the proximal property $\mathrm{(\ref{eq1.3a})}$ for $\mu$-a.e. $\sigma\in\varSigma_{\bK}^+$.

This thus completes the proof of Proposition~\ref{prop3.6}.
\end{proof}

This result shows that relative to each ergodic measure which has at least one nonzero Lyapunov exponent, the set of fiber-chaotic switching laws for System~(\ref{eq1.1}) is ignorable from the viewpoint of ergodic theory.

Let $a_n=n^{(-1)^n}$ for $n=1,2,\dotsc$; then we have
$$
\liminf_{n\to\infty}a_n=0,\quad\limsup_{n\to\infty}a_n=\infty,\quad \textrm{but }\lim_{n\to\infty}\frac{1}{n}\log a_n=0.
$$
Because of this reason, we cannot assert the statement of Proposition~\ref{prop3.6} in the case that System~(\ref{eq1.1}) has only zero Lyapunov exponent at some ergodic measure $\mu$, even if for $d=1$.

From the above arguments, we can easily obtain the following:

\begin{cor}
Let $\sigma\in\varSigma_{\bK}^+$ be fiber-chaotic for System~$(\ref{eq1.1})$. Then either $\sigma$ is a Lyapunov irregular point or System~$(\ref{eq1.1})$ has only zero exponents at $\sigma$.
\end{cor}

As a simple result of \cite[Corollary~2.7]{Dai-LAA} and the multiplicative ergodic theorem, there holds the following lemma.

\begin{lem}\label{lem3.8}
For System~$(\ref{eq1.1})$, $\log\widehat{\pmb{\rho}}$ is a Lyapunov exponent at some ergodic probability measure $\mu$ of $\theta^+$ on $\varSigma_{\bK}^+$. If $S_1,\dotsc,S_K$ all are nonsingular, then $\log\widehat{\pmb{\rho}}_{\mathrm{co}}^{}$ is also a Lyapunov exponent at some ergodic probability measure $\nu$.
Moreover, $\log\widehat{\pmb{\rho}}$ is just the maximal Lyapunov exponent and $\log\widehat{\pmb{\rho}}_{\mathrm{co}}^{}$ is just the minimal Lyapunov exponent of System~$(\ref{eq1.1})$.
\end{lem}

\begin{proof}
We only need to prove that if $S_1,\dotsc,S_K$ all are nonsingular, then $\log\widehat{\pmb{\rho}}_{\mathrm{co}}^{}$ is a Lyapunov exponent at some ergodic probability measure $\nu$ and it is just the minimal Lyapunov exponent of System~$(\ref{eq1.1})$.

Since $\widehat{\pmb{\rho}}_{\mathrm{co}}^{-1}>0$ is just the joint spectral radius of $\{S_1^{-1},\dotsc,S_K^{-1}\}$, from \cite[Corollary~2.7]{Dai-LAA} it follows that $-\log\widehat{\pmb{\rho}}_{\mathrm{co}}^{}$ is the maximal Lyapunov exponent of System~$(\ref{eq1.1})^*$ that is induced by $\{S_1^{-1},\dotsc,S_K^{-1}\}$ and there is an ergodic probability measure $\nu$ of $(\varSigma_{\bK}^+,\theta^+)$ such that System~$(\ref{eq1.1})^*$ has the Lyapunov exponent $-\log\widehat{\pmb{\rho}}_{\mathrm{co}}^{}$ at $\nu$.

Let $\varSigma_{\bK}=\bK^\mathbb{Z}$ be the bi-sided symbolic space whose topology is induced by the cylinders 
\begin{equation*}
[j_m,\dotsc,j_n]=\{\sigma\in\varSigma_{\bK}\colon\sigma(i)=j_i\textrm{ for }m\le i\le n\}
\end{equation*} 
for all $-\infty<m\le n<+\infty$, and let $\theta\colon (\sigma(n))_{n\in\mathbb{Z}}\mapsto(\sigma(n+1))_{n\in\mathbb{Z}}$ be the bi-sided shift transformation on $\varSigma_{\bK}$, which is a homeomorphism. Let $\pi^+\colon(\sigma(n))_{n\in\mathbb{Z}}\mapsto(\sigma(n))_{n\in\mathbb{N}}$ be the natural continuous projection from $\varSigma_{\bK}$ onto $\varSigma_{\bK}^+$. Then $\theta^+$ is a factor of $\theta$, i.e., $\pi^+\circ\theta=\theta^+\circ\pi^+$, such that for any ergodic probability measure $\mu$ of $(\varSigma_{\bK}^+,\theta^+)$ there exists a unique ergodic probability measure $\hat{\mu}$ of $(\varSigma_{\bK},\theta)$ with the property $\mu=\hat{\mu}\circ(\pi^+)^{-1}$ and there are the same Lyapunov exponents at $\mu$ and $\hat{\mu}$ for $\{S_1^{-1},\dotsc,S_K^{-1}\}$ (cf.~\cite[Lemma~6.4]{Dai-JDE07}).

Finally from the multiplicative ergodic theorem for invertible linear cocycle (see e.g. \cite[Theorem~10.3]{Wal}), we can see that $\log\widehat{\pmb{\rho}}_{\mathrm{co}}^{}$ is the Lyapunov exponent of System~(\ref{eq1.1}) at $\nu$ and it is just the minimal Lyapunov exponent of System~$(\ref{eq1.1})$, by considering the negatively direct Lyapunov exponents as $n\to-\infty$.

This completes the proof of Lemma~\ref{lem3.8}.
\end{proof}
\section{Fiber-chaotic dynamical behavior}\label{sec4}
This section will be devoted to proving our main results Theorems~\ref{thm1.2} and \ref{thm1.3} stated in the Introduction.
\subsection{Proof of Theorem~\ref{thm1.2}}\label{sec4.1}
Let $S_1,\dotsc,S_K$ be arbitrarily given real, $d$-by-$d$, nonsingular matrices.
For System (\ref{eq1.1}), let $\pmb{\varLambda}$ be the set that consists of the switching laws $\sigma\in\varSigma_{\bK}^+$ such that the following \textit{uniform fiber-chaoticity} holds:
\begin{equation*}
\liminf_{n\to+\infty}\|S_{\sigma(n)}\dotsm S_{\sigma(1)}\|=0\quad \textrm{and}\quad
\limsup_{n\to+\infty}\|S_{\sigma(n)}\dotsm S_{\sigma(1)}\|_{\textrm{co}}=\infty.
\end{equation*}
Then each $\sigma\in\pmb{\varLambda}$ is fiber-chaotic for System (\ref{eq1.1}) in the sense of Definition~\ref{def1.1}.

To prove one of our main results Theorem~\ref{thm1.2}, we first need the following lemma.

\begin{lem}\label{lem4.1}
Under the same hypotheses of Theorem~\ref{thm1.2}, $\pmb{\varLambda}$ is a dense subset of the topological space $\varSigma_{\bK}^+$.
\end{lem}

\begin{proof}
Let
$(i_1,\dotsc,i_m)\in\bK^m$ and $(j_1,\dotsc,j_n)\in\bK^n$
be two finite-length words such that
\begin{equation*}
\|S_{i_m}\dotsm S_{i_1}\|<1<\|S_{j_n}\dotsm S_{j_1}\|_{\mathrm{co}}.
\end{equation*}
We simply write
\begin{gather*}
\bi=(i_1,\dotsc,i_m),\quad \bi^k=(\stackrel{k\textrm{-folds}}{\overbrace{\bi,\dotsc,\bi}})\in\bK^{km},\\
\bj=(j_1,\dotsc,j_n),\quad \bj^k=(\stackrel{k\textrm{-folds}}{\overbrace{\bj,\dotsc,\bj}})\in\bK^{kn},\\
\intertext{and}
S(\bi)=S_{i_m}\dotsc S_{i_1},\quad S(\bj)=S_{j_n}\dotsc S_{j_1}.
\end{gather*}
Let $\bsigma=(\bsigma(1),\bsigma(2),\dotsc)\in\varSigma_{\bK}^+$ and $\epsilon>0$ be arbitrarily given.
Then one can find an integer $N\ge1$ such that for any
$\sigma\in\varSigma_{\bK}^+$, if $\sigma(1)=\bsigma(1),\dotsc,\sigma(N)=\bsigma(N)$, then the distance $d(\bsigma,\sigma)<\epsilon$.
Set
\begin{equation*}
A=S_{\bsigma(N)}\dotsc S_{\bsigma(1)}.
\end{equation*}
Next, we will construct a fiber-chaotic switching law $\sigma\in\varSigma_{\bK}^+$ for System (\ref{eq1.1}) with $d(\bsigma,\sigma)<\epsilon$.

Since all the matrices $S_1,\dotsc,S_K$ are nonsingular, from the inequalities $\|BA\|\le\|B\|\cdot\|A\|$ and $\|BA\|_{\mathrm{co}}\ge\|B\|_{\mathrm{co}}\cdot\|A\|_{\mathrm{co}}$ we can choose positive integers $\ell_k<L_k$, for $k=1,2,\dotsc$, such that
\begin{gather*}
\|S(\bi)^{\ell_1}A\|<1,\\
\|S(\bj)^{L_1}S(\bi)^{\ell_1}A\|_{\textrm{co}}>1;\\
\|S(\bi)^{\ell_2}S(\bj)^{L_1}S(\bi)^{\ell_1}A\|<\frac{1}{2},\\
\|S(\bj)^{L_2}S(\bi)^{\ell_2}S(\bj)^{L_1}S(\bi)^{\ell_1}A\|_{\textrm{co}}>2;\\
\vdots\quad\vdots\quad\vdots\\
\|S(\bi)^{\ell_k}S(\bj)^{L_{k-1}}\dotsm S(\bj)^{L_1}S(\bi)^{\ell_1}A\|<\frac{1}{k},\\
\|S(\bj)^{L_k}S(\bi)^{\ell_k}S(\bi)^{L_{k-1}}\dotsm S(\bj)^{L_1}S(\bi)^{\ell_1}A\|_{\textrm{co}}>k;\\
\vdots\quad\vdots\quad\vdots\ .
\end{gather*}
Now it is easy to see that the switching law $\sigma$ defined by
\begin{equation*}
\sigma=\left(\bsigma(1),\dotsc,\bsigma(N), \bi^{\ell_1},\bj^{L_1}, \bi^{\ell_2},\bj^{L_2},\bi^{\ell_3},\bj^{L_3},\dotsc\right)
\end{equation*}
is fiber-chaotic for System (\ref{eq1.1}) in the sense of Definition~\ref{def1.1} such that $d(\bsigma,\sigma)<\epsilon$.

This completes the proof of Lemma~\ref{lem4.1}.
\end{proof}

Next, we will prove that $\pmb{\varLambda}$ is a $G_\delta$ subset of $\varSigma_{\bK}^+$; that is, $\pmb{\varLambda}$ is the intersection of a countable collection of open sets of the space $\varSigma_{\bK}^+$.

\begin{lem}\label{lem4.2}
For System $(\ref{eq1.1})$, $\pmb{\varLambda}$ is a $G_\delta$ subset of the topological space $\varSigma_{\bK}^+$.
\end{lem}

\begin{proof}
For any positive integer $i$, let
\begin{equation*}
\varLambda_i^s=\left\{\sigma\in\varSigma_{\bK}^+\colon\forall n_0\in\mathbb{N}, \exists n>n_0\textrm{ with }\|S_{\sigma(n)}\dotsm S_{\sigma(1)}\|<\frac{1}{i}\right\}.
\end{equation*}
Then
\begin{equation*}
\varLambda_i^s=\bigcap_{n_0=1}^\infty\bigcup_{n>n_0}\left\{\sigma\in\varSigma_{\bK}^+\colon\|S_{\sigma(n)}\dotsm S_{\sigma(1)}\|<\frac{1}{i}\right\}.
\end{equation*}
Since $\left\{\sigma\in\varSigma_{\bK}^+\colon\|S_{\sigma(n)}\dotsm S_{\sigma(1)}\|<\frac{1}{i}\right\}$ is open in $\varSigma_{\bK}^+$ for every integer $i>0$ (noting that the cylinder set
\begin{equation*}
[j_1,\dotsc,j_n]=\left\{\sigma\in\varSigma_{\bK}^+\colon\sigma(1)=j_1,\dotsc,\sigma(n)=j_n\right\}
\end{equation*}
is an open subset of the topological space $\varSigma_{\bK}^+$), $\varLambda_i^s$ is a $G_\delta$ set in $\varSigma_{\bK}^+$. Thus,
\begin{equation*}
\pmb{\varLambda}^s:=\bigcap_{i=1}^\infty\varLambda_i^s
\end{equation*}
is also a $G_\delta$ set in $\varSigma_{\bK}^+$.
On the other hand, let
\begin{equation*}
\varLambda_i^u=\left\{\sigma\in\varSigma_{\bK}^+\colon\forall n_0\in\mathbb{N}, \exists n>n_0\textrm{ with }\|S_{\sigma(n)}\dotsm S_{\sigma(1)}\|_{\textrm{co}}>i\right\}.
\end{equation*}
Then the set
\begin{equation*}
\varLambda_i^u:=\bigcap_{n_0=1}^\infty\bigcup_{n>n_0}\left\{\sigma\in\varSigma_{\bK}^+\colon\|S_{\sigma(n)}\dotsm S_{\sigma(1)}\|_{\textrm{co}}>i\right\}.
\end{equation*}
is a $G_\delta$ set in $\varSigma_{\bK}^+$.
Moreover
\begin{equation*}
\pmb{\varLambda}^u:=\bigcap_{i=1}^\infty\varLambda_i^u
\end{equation*}
is a $G_\delta$ set.
Therefore, $\pmb{\varLambda}=\pmb{\varLambda}^s\cap\pmb{\varLambda}^u$ is a $G_\delta$ subset of $\varSigma_{\bK}^+$.

This completes the proof of Lemma~\ref{lem4.2}.
\end{proof}

If there is no additional condition, it may occur that $\pmb{\varLambda}=\varnothing$.

Based on Lemmas~\ref{lem4.1} and \ref{lem4.2} we are now ready to finish the proof of Theorem~\ref{thm1.2}. In fact, we can obtain a result which is stronger than Theorem~\ref{thm1.2}.

\begin{theorem}\label{thm4.3}
Let $S_1,\dotsc,S_K$ all be nonsingular $d\times d$ real matrices. If there exists two finite-length words, say
$(i_1,\dotsc,i_m)\in\bK^m$ and $(j_1,\dotsc,j_n)\in\bK^n$,
such that $\|S_{i_m}\dotsm S_{i_1}\|<1<\|S_{j_n}\dotsm S_{j_1}\|_{\mathrm{co}}$.
then for any $\sigma\in\pmb{\varLambda}$,
\begin{equation*}
\liminf_{n\to+\infty}\|S_{\sigma(n)}\dotsm S_{\sigma(1)}\|=0\quad\textrm{and}\quad
\limsup_{n\to+\infty}\|S_{\sigma(n)}\dotsm S_{\sigma(1)}\|_{\mathrm{co}}=\infty.
\end{equation*}
\end{theorem}

\begin{proof}
We easily see that $\pmb{\varLambda}$ is a dense $G_\delta$ subset of $\varSigma_{\bK}^+$ from Lemmas~\ref{lem4.1} and \ref{lem4.2}. Since each $\sigma\in\pmb{\varLambda}$ is fiber-chaotic for System (\ref{def1.1}), the set of all fiber-chaotic laws of System (\ref{eq1.1}) is residual. This proves Theorem~\ref{thm4.3}.
\end{proof}

Let us consider a simple example in accordance with Theorem~\ref{thm1.2}.

\begin{example}
Given any two constants $\alpha,\beta$ such that $|\alpha|<1$ and $|\beta|>1$, let
\begin{equation*}
S_1=\alpha\left(\begin{matrix}1&1\\0&1\end{matrix}\right)\quad \textrm{and}\quad S_2=\beta\left(\begin{matrix}1&0\\1&1\end{matrix}\right).
\end{equation*}
Then from Theorem~\ref{thm1.2}, it follows that System (\ref{eq1.1}) generated by $S_1$ and $S_2$ is fiber-chaotic in the sense of Definition~\ref{def1.1}.
\end{example}
\subsection{Proof of Theorem~\ref{thm1.3} and pointwise fiber chaos}\label{sec4.2}
For now, we can prove the another main result Theorem~\ref{thm1.3} by improving the argument of Lemma~\ref{lem4.1} as follows:

\begin{proof}[Proof of Theorem~\ref{thm1.3}]
First let the hypotheses of Theorem~\ref{thm1.3} hold. Without loss of generality, let $\|\bx_0\|=\|\by_0\|=1$, and we write $\sigma_0=(\sigma_0(n))_{n\ge1}$ and $\sigma_1=(\sigma_1(n))_{n\ge1}$.

Let $\sigma=(\sigma(n))_{n\ge1}\in\varSigma_{\bK}^+$, $u\in\mathbb{R}^2\setminus\{0\}$ and $N_0\ge1$ be arbitrarily given.
Put
\begin{equation*}
u_{N_0}=S_{\sigma(N_0)}\dotsm S_{\sigma(1)}u.
\end{equation*}
Since $S_1,\dotsc,S_K$ all are nonsingular, we see $u_{N_0}\not=0$.

Step 1). Let $\varepsilon_1>0$ be sufficiently small. Because $S_1$ is an irrational rotation, there exists an integer $N_1\ge N_0+1$ such that
\begin{equation*}
\big{\|}u_{N_1}-\|u_{N_0}\|\bx_0\big{\|}<\varepsilon_1,\quad \textrm{where }u_{N_1}=S_1^{N_1-N_0}u_{N_0}.
\end{equation*}
Further there exists an integer $N_1^\prime\ge N_1+1$ such that
\begin{equation*}
\|u_{N_1^\prime}\|<\varepsilon_1,\quad \textrm{where }u_{N_1^\prime}=S_{\sigma_0(N_1^\prime-N_1)}\dotsm S_{\sigma_0(1)}u_{N_1}.
\end{equation*}

Step 2). Let $\varepsilon_2>0$ with $\varepsilon_2<\varepsilon_1$. Since $u_{N_1^\prime}\not=0$, similarly we can find some $N_2\ge N_1^\prime+1$ such that
\begin{equation*}
\big{\|}u_{N_2}-\|u_{N_1^\prime}\|\by_0\big{\|}<\varepsilon_2,\quad \textrm{where }u_{N_2}=S_1^{N_2-N_1^\prime}u_{N_1^\prime}.
\end{equation*}
Further there exists an integer $N_2^\prime\ge N_2+1$ such that
\begin{equation*}
\|u_{N_2^\prime}\|>\frac{1}{\varepsilon_2},\quad \textrm{where }u_{N_2^\prime}=S_{\sigma_1(N_2^\prime-N_2)}\dotsm S_{\sigma_1(1)}u_{N_2}.
\end{equation*}

Step 3). Taking $\varepsilon_1>\varepsilon_2>\varepsilon_3>\dotsm\rightarrow0$ and repeating the above Step 1) and Step 2) for $\varepsilon_3,\varepsilon_4, \dotsc$, we can choose a switching law $\sigma^\prime\in\varSigma_{\bK}^+$ such that
\begin{gather*}
\sigma(1)=\sigma^\prime(1), \dotsc, \sigma(N_0)=\sigma^\prime(N_0);\\
\sigma(N_0+1)=\dotsm=\sigma(N_1)=1;\\
\sigma(N_1+1)=\sigma_0(1),\dotsc,\sigma(N_1^\prime)=\sigma_0(N_1^\prime-N_1); \\
\sigma(N_1^\prime+1)=\dotsm=\sigma(N_2)=1;\\
\sigma(N_2+1)=\sigma_1(1),\dotsc,\sigma(N_2^\prime)=\sigma_1(N_2^\prime-N_2);\\
\vdots\quad\vdots\quad\vdots
\end{gather*}
and
\begin{equation*}
\liminf_{n\to\infty}\|x_n(u,\sigma^\prime)\|=0\quad\textrm{and}\quad\limsup_{n\to\infty}\|x_n(u,\sigma^\prime)\|=\infty.
\end{equation*}
This completes the proof of Theorem~\ref{thm1.3}.
\end{proof}

To counter Theorem~\ref{thm1.3}, we now consider another simple example .

\begin{example}
Let $\alpha\in(0,1)$ be an irrational number and
\begin{equation*}
S_1=\left(\begin{matrix}\cos2\pi\alpha&\sin2\pi\alpha\\-\sin2\pi\alpha&\cos2\pi\alpha\end{matrix}\right)\quad \textrm{and}\quad S_2=\beta\left(\begin{matrix}2&0\\-1&\frac{1}{2}\end{matrix}\right).
\end{equation*}
Then from Theorem~\ref{thm1.3}, it follows that System (\ref{eq1.1}) generated by $S_1$ and $S_2$ has fiber-chaotic switching laws.
\end{example}

To System $(\ref{eq1.1})$ and for any nonzero $x_0\in\mathbb{R}^d$, we now define the \textit{pointwise fiber-chaotic switching laws} as follows:
\begin{equation*}
\pmb{\varLambda}(x_0)=\left\{\sigma\colon\liminf_{n\to\infty}\|x_n(x_0,\sigma)\|=0,\ \limsup_{n\to\infty}\|x_n(x_0,\sigma)\|=\infty\right\}.
\end{equation*}

Then similar to Lemma~\ref{lem4.2}, we can easily obtain the following simple lemma.

\begin{lem}\label{lem4.6}
For System $(\ref{eq1.1})$, the set $\pmb{\varLambda}(x_0)$
is a $G_\delta$-subset of $\varSigma_{\bK}^+$ for any nonzero $x_0\in\mathbb{R}^d$.
\end{lem}

\begin{proof}
Let $x_0\in\mathbb{R}^d\setminus\{0\}$. For any positive integer $i$, let
\begin{equation*}
\varLambda_i^s(x_0)=\{\sigma\colon\forall n_0\in\mathbb{N}, \exists n>n_0\textrm{ with }\|S_{\sigma(n)}\dotsm S_{\sigma(1)}x_0\|<i^{-1}\}.
\end{equation*}
Then $\varLambda_i^s(x_0)$ and $\pmb{\varLambda}^s(x_0)=\cap_{i=1}^\infty\varLambda_i^s(x_0)$ both are $G_\delta$-sets in $\varSigma_{\bK}^+$.
Let
\begin{equation*}
\varLambda_i^u(x_0)=\{\sigma\colon\forall n_0\in\mathbb{N}, \exists n>n_0\textrm{ with }\|S_{\sigma(n)}\dotsm S_{\sigma(1)}x_0\|>i\}.
\end{equation*}
Then $\varLambda_i^u(x_0)$ and $\pmb{\varLambda}^u(x_0)=\cap_{i=1}^\infty\varLambda_i^u(x_0)$ both are $G_\delta$-sets in $\varSigma_{\bK}^+$.
Thus $\pmb{\varLambda}(x_0)=\pmb{\varLambda}^s(x_0)\cap\pmb{\varLambda}^u(x_0)$ is a $G_\delta$-sets in $\varSigma_{\bK}^+$. This proves Lemma~\ref{lem4.6}.
\end{proof}

From Theorem~\ref{thm1.3} and Lemma~\ref{lem4.6}, we can obtain the following.

\begin{theorem}\label{thm4.7}
Under the hypotheses of Theorem~\ref{thm1.3}, for any nonzero initial state $x_0\in\mathbb{R}^2$, $\pmb{\varLambda}(x_0)$ is residual in the space $\varSigma_{\bK}^+$.
\end{theorem}

\begin{proof}
Let $x_0\in\mathbb{R}^2\setminus\{0\}$ be arbitrary. From Lemma~\ref{lem4.6}, we see that $\pmb{\varLambda}(x_0)$
is a $G_\delta$ subset of $\varSigma_{\bK}^+$. By Theorem~\ref{thm1.3}, $\pmb{\varLambda}(x_0)$
is dense in $\varSigma_{\bK}^+$. Thus this completes the proof of Theorem~\ref{thm4.7}.
\end{proof}

More generally, for any nonempty set $X\subseteq\mathbb{R}^d\setminus\{0\}$ we define the set
\begin{equation*}
\pmb{\varLambda}(X)=\left\{\sigma\in\varSigma_{\bK}^+\colon\liminf_{n\to\infty}\|x_n(x_0,\sigma)\|=0\textrm{ and } \limsup_{n\to\infty}\|x_n(x_0,\sigma)\|=\infty,\ \forall x_0\in X\right\}.
\end{equation*}
Then in the case of Theorem~\ref{thm1.2}, $\pmb{\varLambda}(\mathbb{R}^d\setminus\{0\})$ is residual in $\varSigma_{\bK}^+$. However, for the case of Theorem~\ref{thm1.3}, we can only obtain the following weak result from Theorem~\ref{thm4.7}:

\begin{cor}
Under the hypotheses of Theorem~\ref{thm1.3}, there exists a dense countable subset $X$ of $\mathbb{S}^1$ such that $\pmb{\varLambda}(X)$ is residual in $\varSigma_{\bK}^+$.
\end{cor}

Here $\mathbb{S}^1=\{x_0\in\mathbb{R}^2\colon\|x_0\|=1\}$ is the unit circle in $\mathbb{R}^2$.
\subsection{Completely product (co-) unbounded systems}
We now turn to another basic property of fiber-chaotic systems.

\begin{defn}\label{def4.9}
System $(\ref{eq1.1})$ is called \textit{completely product unbounded} if restricted to every nonempty, common and invariant subspace of $\mathbb{R}^d$, it is product unbounded.
\end{defn}

So, if System $(\ref{eq1.1})$ is completely product unbounded then it is product unbounded. But the converse is not necessarily true. For example, for
\begin{equation*}
S_1=\left(\begin{matrix}1&1\\0&1\end{matrix}\right)\quad \textrm{and}\quad S_2=\left(\begin{matrix}1&1\\0&1\end{matrix}\right),
\end{equation*}
System $(\ref{eq1.1})$ is product unbounded but not completely product unbounded.

We will employ the following simple fact in the next section.

\begin{lem}\label{lem4.10}
If System $(\ref{eq1.1})$ has a fiber-chaotic switching law, then it is completely product unbounded.
\end{lem}

\begin{proof}
This statement follows immediately from the Definitions~\ref{def1.1} and \ref{def4.9}.
\end{proof}

Dually, we can obtain the followings:\begin{defn}\label{def4.11}
System $(\ref{eq1.1})$ is called \textit{completely product co-unbounded} if restricted to every nonempty, common and invariant subspace of $\mathbb{R}^d$, it is product co-unbounded.
\end{defn}

\begin{lem}\label{lem4.12}
If System $(\ref{eq1.1})$ has a fiber-chaotic switching law, then it is completely product co-unbounded.
\end{lem}

Finally, it should be noted here that the nonsingularity of the constituent subsystems of System~(\ref{eq1.1}) is important for the generic fiber-chaos property. For example, if
\begin{equation*}
S_1=\left(\begin{matrix}1&0\\0&0\end{matrix}\right),\quad S_2=\left(\begin{matrix}0&0\\0&1\end{matrix}\right),\quad S_3=\left(\begin{matrix}\frac{1}{2}&0\\0&\frac{1}{2}\end{matrix}\right),
\quad
\textrm{and}\quad
S_4=\left(\begin{matrix}2&0\\0&2\end{matrix}\right),
\end{equation*}
then $\|S_3\|=\frac{1}{2}<1<2=\|S_4\|_{\mathrm{co}}$ but there does not exist any fiber-chaotic switching laws for System~(\ref{eq1.1}) in the open cylinder set $[1,2]=\{\sigma\in\varSigma_{\bK}^+\colon\sigma(1)=1,\sigma(2)=2\}$.
\section[Coexistence and Periodical stability]{Coexistence of positive, zero and negative Lyapunov exponents and periodical stability}\label{sec5}
In this section, we will prove some necessary conditions for the fiber chaos and show that every periodically stable system has no fiber-chaotic dynamics.

\subsection{Nonexistence of fiber-chaotic switching laws}
The following result is our basic tool for proving the non-fiber-chaotic dynamics in our sense of Definition~\ref{def1.1}.

\begin{theorem}\label{thm5.1}
The following two statements are satisfied for System $(\ref{eq1.1})$:
\begin{enumerate}
\item[$(1)$] If $\widehat{\pmb{\rho}}=1$, then there does not exist any fiber-chaotic switching laws.
\item[$(2)$] If $\widehat{\pmb{\rho}}_{\mathrm{co}}^{}=1$,
then there does not exist any fiber-chaotic switching laws.
\end{enumerate}
\end{theorem}

\begin{proof}
Let $\widehat{\pmb{\rho}}=1$. We first note that according to the Definition~\ref{def1.1} before, if System (\ref{eq1.1}) is product bounded (cf.~Sect.~\ref{sec2.2}), then it does not have any fiber-chaotic switching laws.

By contradiction, we let $\bsigma\in\varSigma_{\bK}^+$ be fiber-chaotic for System (\ref{eq1.1}) in the sense of Definition~\ref{def1.1}; and then System (\ref{eq1.1}) is completely product unbounded (cf.~Definition~\ref{def4.9}) from Lemma~\ref{lem4.10}.

From Elsner's reduction theorem (cf.~Sect.~\ref{sec2.2}), there is no loss of generality in assuming
\begin{equation*}
S_k=\begin{pmatrix}S_k^{(1)}&*_k\\0&D_k^{(1)}\end{pmatrix},\quad k=1,\dotsc,K,
\end{equation*}
such that
\begin{equation*}
S_k^{(1)}\in\mathbb{R}^{d_1\times d_1}\quad\textrm{and}\quad D_k^{(1)}\in\mathbb{R}^{(d-d_1)\times (d-d_1)},\quad k=1,\dotsc,K,
\end{equation*}
for some integer $1\le d_1<d$, and such that the inclusion system based on $\left\{S_1^{(1)}, \dotsc, S_K^{(1)}\right\}$ has the joint spectral radius $1$ and moreover, $\pmb{\sigma}$ is a fiber-chaotic switching law for it as well. Repeating this argument finite times, it follows that one can find $K$ nonsingular $1\times1$ matrices
$S_1^{(r)},S_2^{(r)},\dotsc, S_K^{(r)}$ such that their induced inclusion system has the joint spectral radius $1$ and the fiber-chaotic switching law $\bsigma$. This is a contradiction to Elsner's reduction theorem. This proves the statement (1) of Theorem~\ref{thm5.1}.

For now, let $\widehat{\pmb{\rho}}_{\mathrm{co}}^{}=1$. Similarly if System (\ref{eq1.1}) is product co-bounded (cf.~Definition~\ref{def2.4}), then it does not have any fiber-chaotic switching laws from Lemma~\ref{lem4.12}. The rest of the proof is similar to that of the statement (1) using Lemmas~\ref{lem2.3} and \ref{lem2.5}. So we omit the details.

This thus completes the proof of Theorem~\ref{thm5.1}.
\end{proof}

This result shows that our Definition~\ref{def1.1} is essentially different with Definition~\ref{def3.1} of Balde and Jouan. The following two results Theorems~\ref{thm5.2} and \ref{thm5.3} are simple consequences of Theorem~\ref{thm5.1}.

\subsection{Necessary conditions for fiber chaos}\label{sec5.2}
The following theorem shows that fiber chaos implies the coexistence of positive, zero and negative Lyapunov exponents for System~(\ref{eq1.1}).

\begin{theorem}\label{thm5.2}
Let $S_1,\dotsc,S_K$ be nonsingular. If System $(\ref{eq1.1})$ has a fiber-chaotic switching law, then there hold the following statements.
\begin{enumerate}
\item[$(1)$] System $(\ref{eq1.1})$ has at least one positive Lyapunov exponent, i.e, $\exists(x_0,\sigma)\in(\mathbb{R}^d\setminus\{0\})\times\varSigma_{\bK}^+$ such that $\lambda(x_0,\sigma)>0$.

\item[$(2)$] System $(\ref{eq1.1})$ has at least one negative Lyapunov exponent, i.e, $\exists(x_0,\sigma)\in(\mathbb{R}^d\setminus\{0\})\times\varSigma_{\bK}^+$ such that $\lambda(x_0,\sigma)<0$.

\item[$(3)$] For any $x_0\in\mathbb{R}^d\setminus\{0\}$, there exists some $\sigma\in\varSigma_{\bK}^+$ such that System $(\ref{eq1.1})$ has the zero Lyapunov exponent at the data $(x_0,\sigma)$.
\end{enumerate}
\end{theorem}

\begin{proof}
We will proceed the proof of (1) by contradiction. Assume for any nonzero $x_0\in\mathbb{R}^d$ and any $\sigma\in\varSigma_{\bK}^+$ we have only the non-positive Lyapunov exponents
$$
\lim_{n\to\infty}\frac{1}{n}\log\|x_n(x_0,\sigma)\|\le0\quad (\textrm{if the limit exists here}).
$$
Then from Lemma~\ref{lem3.8}, it follows that the joint spectral radius $\widehat{\pmb{\rho}}$ of System (\ref{eq1.1}) is less than or equal to $1$. Thus a contradiction comes from Theorem~\ref{thm5.1}.\,(1). This thus proves the statement (1).

The statement (2) may follows similarly. Assume for any nonzero $x_0\in\mathbb{R}^d$ and any $\sigma\in\varSigma_{\bK}^+$ we have only the nonnegative Lyapunov exponents
$$
\lim_{n\to\infty}\frac{1}{n}\log\|x_n(x_0,\sigma)\|\ge0\quad (\textrm{if the limit exists here}).
$$
Then from Lemma~\ref{lem3.8}, it follows that the joint spectral co-radius $\widehat{\pmb{\rho}}_{\mathrm{co}}^{}$ of System (\ref{eq1.1}) is bigger than or equal to $1$. If $\widehat{\pmb{\rho}}_{\mathrm{co}}^{}>1$, then System~(\ref{eq1.1}) is uniformly asymptotically divergent and this is a contradiction.
If $\widehat{\pmb{\rho}}_{\mathrm{co}}^{}=1$, then a contradiction comes from Theorem~\ref{thm5.1}.\,(2). This proves the statement (2).

To prove the statement (3), let $\bsigma=(\bsigma(n))_{n\ge1}$ be a fiber-chaotic switching law of System $(\ref{eq1.1})$. Let $\mathbb{S}^{d-1}$ be the unit sphere of the state space $\mathbb{R}^d$. Let $0<\alpha<1<\beta<\infty$ be such that $\alpha<\|S_k\|_{\mathrm{co}}\le\|S_k\|<\beta$ for each $1\le k\le K$. Set
\begin{equation*}
\mathbb{T}_{\alpha1}=\{u\in\mathbb{R}^d\colon\alpha\le\|u\|<1\}\quad\textrm{and}\quad\mathbb{T}_{1\beta}=\{u\in\mathbb{R}^d\colon1<\|u\|\le\beta\}.
\end{equation*}

Since $\mathbb{S}^{d-1}\cup\mathbb{T}_{1\beta}$ is compact and $x_n(x_0,\bsigma)$ is continuous with respect to $x_0\in\mathbb{S}^{d-1}\cup\mathbb{T}_{1\beta}$ for every $n\ge1$, from the fiber-chaoticity of $\bsigma$ it follows that one can find an open cover of $\mathbb{T}_{1\beta}$, say $V_1,\dotsc,V_r$, and positive integers $n_1,\dotsc,n_r$ such that
\begin{equation*}
\alpha<\|x_{n_i}(y_0,\bsigma)\|<1\quad\textrm{and}\quad1\le\|x_n(y_0,\bsigma)\|\textrm{ for }1\le n<n_i,\quad \forall y_0\in V_i,\ i=1,\dotsc,r.
\end{equation*}
Similarly, one can find an open cover of $\mathbb{T}_{\alpha1}$, say $U_1,\dotsc,U_r$, and positive integers $m_1,\dotsc,m_r$ such that
\begin{equation*}
1<\|x_{m_i}(y_0,\bsigma)\|<\beta\quad\textrm{and}\quad \|x_n(y_0,\bsigma)\|\le1\textrm{ for }1\le n<m_i,\quad \forall y_0\in U_i,\ i=1,\dotsc,r.
\end{equation*}
Now, let $x_0\in\mathbb{S}^{d-1}$ be arbitrarily given. We are going to define a switching law $\sigma$ such that
$$\lambda(x_0,\sigma)=\lim_{n\to\infty}\frac{1}{n}\log\|x_n(x_0,\sigma)\|=0$$
as follows: For some $\ell_1\ge1$, let $\sigma(1), \dotsc, \sigma(\ell_1)$ are defined such that $x_1=S_{\sigma(1)}x_0, \dotsc, x_{\ell_1-1}=S_{\sigma(\ell_1)}x_{\ell_1-2}$ all belong to $\mathbb{S}^{d-1}$ and $x_{\ell_1}=S_{\sigma(\ell_1)}x_{\ell_1-1}$ belongs to $\mathbb{T}_{\alpha1}\cup\mathbb{T}_{1\beta}$. (Note here that we can always find such an $\ell_1$ by the fiber-chaoticity of $\bsigma$ for System~(\ref{eq1.1}).)

If $x_{\ell_1}$ belongs to $\mathbb{T}_{1\beta}$ and without loss of generality let $x_{\ell_1}\in V_{i_1}$; then $\ell_2=n_{i_1}$ and define
\begin{equation*}
\sigma(\ell_1+1)=\bsigma(1),\ \dotsc,\ \sigma(\ell_1+n_{i_1})=\bsigma(n_{i_1}).
\end{equation*}
If $x_{\ell_1}$ belongs to $\mathbb{T}_{\alpha1}$ and without loss of generality let $x_{\ell_1}\in U_{i_1}$; then $\ell_2=m_{i_1}$ and define
\begin{equation*}
\sigma(\ell_1+1)=\bsigma(1),\ \dotsc,\ \sigma(\ell_1+m_{i_1})=\bsigma(m_{i_1}).
\end{equation*}
Now for $x_{\ell_1+\ell_2}=S_{\sigma(\ell_1+\ell_2)}\dotsm S_{\sigma(\ell_1+1)}x_{\ell_1}\in\mathbb{T}_{\alpha1}\cup\mathbb{S}^{d-1}\cup\mathbb{T}_{1\beta}$, repeating the above construction, we can define the desired switching law $\sigma$ with $\lambda(x_0,\sigma)=0$, because $\alpha^{\max\{m_1,\dotsc,m_r\}/2}\le\|x_n(x_0,\sigma)\|\le\beta^{\max\{n_1,\dotsc,n_r\}/2}$ for all $n\ge1$.

This thus completes the proof of Theorem~\ref{thm5.2}.
\end{proof}

The above Theorem~\ref{thm5.2} presents us with some necessary conditions for the fiber chaos of System (\ref{eq1.1}). However, we should note that these conditions are not sufficient for fiber chaos. For example, let
$$
S_1=\left[\begin{matrix}\frac{1}{2}&0\\0&2\end{matrix}\right]\quad \textrm{and}\quad  S_2=\left[\begin{matrix}1&0\\0&1\end{matrix}\right].
$$
Then System~(\ref{eq1.1}) based on $S_1$ and $S_2$ has simultaneously a positive, zero and negative Lyapunov exponents, but it does not have any fiber-chaotic switching laws.

\subsection{Periodically stable systems}
Recall that System $(\ref{eq1.1})$ described as in Section~\ref{sec1} is called, from e.g. \cite{Gur-LAA95, SWMWK, DHX-aut}, \textit{periodically stable} if for any finite-length words $(k_1,\dotsc,k_\pi)\in\bK^\pi, \pi\ge1$, there holds that the spectral radius $\rho(S_{k_\pi}\dotsm S_{k_1})$ of $S_{k_\pi}\dotsm S_{k_1}$ is less than $1$. Then a periodically stable System $(\ref{eq1.1})$ does not need to be absolutely stable
from \cite{BM, BTV, Koz07, HMST}; but it is almost surely exponentially stable in terms of some ergodic measures, see \cite{DHX-aut} and \cite[Theorem~C$^\prime$]{Dai-JDE}.

The following Theorem~\ref{thm5.3} shows that a periodically stable linear inclusion system
is ``simple'' from our viewpoint of fiber-chaos dynamics.

\begin{theorem}\label{thm5.3}
If System $(\ref{eq1.1})$ is periodically stable, then its every switching law is not fiber-chaotic.
\end{theorem}

\begin{proof}
Since System $(\ref{eq1.1})$ is periodically stable, it holds that the joint spectral radius $\widehat{\pmb{\rho}}$ is less than or equal to $1$ from the Berger-Wang spectral formula (cf.~Sect.~\ref{sec2}). If $\widehat{\pmb{\rho}}<1$, then System~(\ref{eq1.1}) is uniformly exponentially stable from \cite{Bar} and so the statement holds. If $\widehat{\pmb{\rho}}=1$, then the statement follows from Theorem~\ref{thm5.1}.\,(1).

This completes the proof of Theorem~\ref{thm5.3}.
\end{proof}

In fact, the following Lemma~\ref{lem5.4} shows a lower dimensional periodically stable system is product bounded.
It is a well-known fact that for System (\ref{eq1.1}), if it holds that $\rho(S_{\sigma(n)}\dotsm S_{\sigma(1)})\le1$ for all $\sigma\in\varSigma_{\bK}^+$, then
\begin{equation*}
\max_{\sigma\in\varSigma_{\bK}^+}\|S_{\sigma(n)}\dotsm S_{\sigma(1)}\|=\mathrm{O}(n^{d-1});
\end{equation*}
see, e.g., \cite{BW, Bel, Lur}. In the periodically stable case (or equivalently, $\rho(S_{\sigma(n)}\dotsm S_{\sigma(1)})<1\;\forall \sigma\in\varSigma_{\bK}^+$ and $n\ge1$), we can get a more subtle estimate as follows.

\begin{lem}\label{lem5.4}
Let System $(\ref{eq1.1})$ be periodically stable with dimension $d\ge2$. Then
\begin{equation*}
\max_{\sigma\in\varSigma_{\bK}^+}\|S_{\sigma(n)}\dotsm S_{\sigma(1)}\|=\mathrm{O}(n^{\lfloor d/2-1\rfloor}).
\end{equation*}
In particular, if $2\le d\le3$ then System $(\ref{eq1.1})$ is product
bounded in $\mathbb{R}^{d\times d}$; if $4\le d\le 5$ then $\|S_{\sigma(n)}\dotsm S_{\sigma(1)}\|$ is at most linearly increasing.
\end{lem}

Here $\lfloor x\rfloor$ represents the largest integer which is not greater than
$x$ for any $x\ge0$ like $\lfloor 0.3\rfloor=0, \lfloor 1.2\rfloor=1$.

\begin{proof}
We will prove the statement by induction on the dimension $d$ of System (\ref{eq1.1}). We first notice that if System (\ref{eq1.1}) is periodically stable with dimension $d=1$, then
there exists a constant $0<\gamma<1$
so that
$$\|S_{\sigma(n)}\dotsm S_{\sigma(1)}\|\le\gamma^n$$
for all $\sigma\in\varSigma_{\bK}^+$ and all $n\ge1$.

Let $m\ge2$ be arbitrarily given. Assume that the statement is true for $d<m$. It suffices to claim that the
statement is also true for $d=m$.

Let $d=m$. The periodical stability of System (\ref{eq1.1}) implies that the joint spectral radius $\pmb{\rho}\le
1$. If System (\ref{eq1.1}) is product bounded then we are done. Otherwise, according to Elsner's reduction theorem we can assume that each $S_k$ has the form
\begin{equation*}
S_k=\begin{pmatrix}S_k^{1,1}& B_k\\ 0& S_k^{2,2}\end{pmatrix},\quad k=1,\dotsc,K,
\end{equation*}
where
$\left\{S_1^{1,1},\dotsc,S_K^{1,1}\right\}\subset\mathbb{R}^{d_1\times
d_1}$, $\{B_1,\dotsc,B_K\}\subset\mathbb{R}^{d_1\times(m-d_1)}$, and
$\left\{S_1^{2,2},\dotsc,S_K^{2,2}\right\}\subset\mathbb{R}^{(m-d_1)\times
(m-d_1)}$ for some $1\leq d_1<m$.
Thus, for any $\sigma\in\varSigma_{\bK}^+$ and $n\ge1$
\begin{equation*}
S_{\sigma(n)}\dotsc S_{\sigma(1)}=\begin{pmatrix}S_{\sigma(n)}^{1,1}\dotsm S_{\sigma(1)}^{1,1}&\spadesuit_{\sigma(n)\dotsm\sigma(1)}\\
0& S_{i_1}^{(2,2)}\cdots
S_{\sigma(n)}^{2,2}\dotsm S_{\sigma(1)}^{2,2}\end{pmatrix}
\end{equation*}
where
\begin{equation*}
\spadesuit_{\sigma(n)\dotsm\sigma(1)}=\sum_{j=1}^{n}S_{\sigma(n)}^{1,1}\dotsm S_{\sigma(j+1)}^{1,1}B_{\sigma(j-1)}S_{\sigma(n)}^{2,2}\dotsc S_{\sigma(1)}^{2,2}.
\end{equation*}
We can choose a constant $C_1>0$ such that
\begin{equation*}
\|B_k\|\le C_1\qquad\forall k=1,\dotsc,K.
\end{equation*}
Now we only need to consider the following two cases.

Case I: When $d_1=1$ or $m-d_1=1$, we can obtain either
\begin{equation*}
\|S_k^{1,1}\|\le\gamma<1\quad\textrm{for }1\le k\le K
\end{equation*}
or
\begin{equation*}
\|S_k^{2,2}\|\le\gamma<1\quad\textrm{for }1\le k\le K,
\end{equation*}
for some constant $0<\gamma<1$. Hence we have
\begin{equation*}
\|\spadesuit_{\sigma(n)\dotsm\sigma(1)}\|\leq
\begin{cases}
C_1C& \textrm{if }m=2;\\
C_1Cn^{\lfloor m/2-1\rfloor}(1+\gamma+\dotsm+\gamma^{n-1})& \textrm{if }m>2,
\end{cases}
\end{equation*}
by the induction assumption, for some constant $C>0$ that is independent of the choices of the switching law $\sigma$. Thus the statement holds in this case.

Case II: When $2\leq d_1<m-1$, according to the induction assumption, it
follows that
\begin{equation*}
\|\spadesuit_{\sigma(n)\dotsm\sigma(1)}\|\leq
C_1Cn^{\lfloor d_1/2-1\rfloor}Cn^{\lfloor(m-d_1)/2-1\rfloor}n
\leq C_1C^2n^{[m/2-1]}.
\end{equation*}
Here we have used the following inequality:
\begin{equation*}
\left\lfloor\frac{d_1}{2}-1\right\rfloor+\left\lfloor\frac{m-d_1}{2}-1\right\rfloor+1\leq \left\lfloor\frac{m}{2}-1\right\rfloor,
\end{equation*}
for any $2\leq d_1<m-1$, which implies the desired result.

This completes the proof of Lemma~\ref{lem5.4}.
\end{proof}

This lemma together with Lemma~\ref{lem3.3} implies the following stability result.

\begin{theorem}\label{thm5.5}
Let System $(\ref{eq1.1})$ be periodically stable with dimension $2\le d\le 3$. Then,
\begin{equation*}
\|S_{\sigma(n)}\dotsm S_{\sigma(1)}\|\to0 \quad \textrm{as }n\to+\infty,
\end{equation*}
for every Balde-Jouan nonchaotic switching laws $\sigma\in\varSigma_{\bK}^+$.
\end{theorem}

\begin{proof}
From Lemma~\ref{lem5.4}, it follows that System (\ref{eq1.1}) is product bounded. So, we can define a norm $\pmb{\|}\cdot\pmb{\|}$ on $\mathbb{R}^{d\times d}$ such that $\pmb{\|}S_k\pmb{\|}\le1$ for all $k=1,\dotsc,K$.

Let $\sigma\in\varSigma_{\bK}^+$ be an arbitrary nonchaotic switching laws of Balde and Jouan as in Definition~\ref{def3.1}. Then we can choose some $\kappa\in\{1,\dotsc,K\}$ as in Lemma~\ref{lem3.3}. Since $\rho(S_\kappa)<1$, we can find some $N>1$ such that $\pmb{\|}S_\kappa^N\pmb{\|}<1$. Then the statement comes from Lemma~\ref{lem3.3} and the sub-multiplicity of matrix norm.
\end{proof}

In fact from \cite{Dai-JDE}, we can see that $\|S_{\sigma(n)}\dotsm S_{\sigma(1)}\|\to0$ exponentially fast as $n\to+\infty$ in Theorem~\ref{thm5.5}.

Extremal norm has been a popular tool for problems of stability and joint spectral radius. Sufficient conditions for the
existence of extremal norms are obtained in \cite{Wir}. As an aside of Lemma~\ref{lem5.4} is the following statement.

\begin{prop}\label{prop5.6}
Let System $(\ref{eq1.1})$ be periodically stable with dimension $2\le d\le 3$. Then, there holds at least one of the following two statements.
\begin{enumerate}
\item[$(1)$] $\mathrm{(Finiteness\ of\ spectrum)}$ There is a word $(k_1,\dotsc,k_\pi)\in\bK^\pi$, for some $\pi\ge1$, such that
\begin{equation*}
\widehat{\pmb{\rho}}=\sqrt[\pi]{\rho(S_{k_\pi}\dotsm S_{k_1})}.
\end{equation*}

\item[$(2)$] $\mathrm{(Finiteness\ of\ norm)}$ There is some $($extremal$)$ norm $\|\cdot\|_*$, defined on $\mathbb{R}^{d\times d}$, such that
\begin{equation*}
\widehat{\pmb{\rho}}=\max_{\sigma\in\varSigma_{\bK}^+}\sqrt[n]{\|S_{\sigma(n)}\dotsm S_{\sigma(1)}\|_*}\quad\forall n\ge1.
\end{equation*}
\end{enumerate}
Here $\widehat{\pmb{\rho}}$ is the joint spectral radius of System $(\ref{eq1.1})$ defined as in Section~\ref{sec2.1}.
\end{prop}

\begin{proof}
If the statement (1) of Proposition~\ref{prop5.6} holds, then we are done. Otherwise, without loss of generality we may assume System
(\ref{eq1.1}) is periodically stable and $\widehat{\pmb{\rho}}=1$ from Berger-Wang's formula. Thus it is product bounded according to Lemma~\ref{lem5.4}.
Then there exists a vector norm $\|\cdot\|_*$ defined on $\mathbb{R}^{d}$, where $2\le d\le 3$, such that $\|S_k\|_*\le 1$ for any $1\le k\le K$. Therefore, one has
$$\sqrt[n]{\|S_{\sigma(n)}\dotsm S_{\sigma(1)}\|_*}\le \widehat{\pmb{\rho}}\quad\forall \sigma\in\varSigma_{\bK}^+\textrm{ and }n\ge1.$$
This implies that
\begin{equation*}
\widehat{\pmb{\rho}}=\inf_{n\ge1}\left\{\max_{\sigma\in\varSigma_{\bK}^+}\sqrt[n]{\|S_{\sigma(n)}\dotsm S_{\sigma(1)}\|_*}\right\}\le\sup_{n\ge1}\left\{\max_{\sigma\in\varSigma_{\bK}^+}\sqrt[n]{\|S_{\sigma(n)}\dotsm S_{\sigma(1)}\|_*}\right\}\le\widehat{\pmb{\rho}},
\end{equation*}
and the proof of Proposition~\ref{prop5.6} is thus completed.
\end{proof}

We note here that this result cannot be proved by directly reducing the dimension of $\mathbb{R}^d$, since an extremal norm of some sub-blocks of System (\ref{eq1.1}) does not need to be an extremal norm for the full dimensional case.

We ends this section with some remarks on Proposition~\ref{prop5.6}.

\begin{remark}
For System (\ref{eq1.1}) in the case of $2\le d\le 3$, if spectral finiteness property does not hold, then there
exists an extremal norm $\|\cdot\|_*$. According to \cite{TB, BT1, BT2, BT3, Koz07} and so on, many systems do not satisfy the spectral finiteness property. Conversely, the non-existence of an extremal norm implies that the finiteness property must hold.
\end{remark}

\begin{remark}
Based on \cite{BM, BTV, Koz07, HMST} we can easily see that Proposition~\ref{prop5.6} does not need to hold for the case $d\ge4$. In fact, there are uncountably many values of the real parameters $\alpha,
\beta$ such that for each pair $(\alpha,\beta)$,
$F=\{F_1,F_2\}$ is periodically stable, where
$$
F_1=\alpha\begin{pmatrix}1&1\\ 0&1\end{pmatrix}\quad \textrm{and} \quad
F_2=\beta\begin{pmatrix}1&0\\ 1&1\end{pmatrix};
$$
but there is at least one
switching law $\pmb{\sigma}\in\varSigma_{\bK}^+$ where $\bK=\{1,2\}$ such that
\begin{equation*}
\|F_{\pmb{\sigma}(n)}\cdots F_{\pmb{\sigma}(1)}\|\not\to 0\quad \textrm{as }n\to+\infty.
\end{equation*}
Define
$$
S_1=\begin{pmatrix}F_1&F_1\\ 0&F_1\end{pmatrix}\quad \textrm{and}\quad S_2=\begin{pmatrix}F_2&F_2\\ 0&F_2\end{pmatrix}.
$$
Then for any $\sigma\in\varSigma_{\bK}^+$ and any $n\ge1$, we have
$$
S_{\sigma(n)}\dotsm S_{\sigma(1)}=\begin{pmatrix}F_{\sigma(n)}\dotsm F_{\sigma(1)}&n F_{\sigma(n)}\dotsm F_{\sigma(1)}\\ 0&F_{\sigma(n)}\dotsm F_{\sigma(1)}\end{pmatrix}.
$$
For $\pmb{\sigma}$, we particularly get
$$\limsup_{n\to+\infty}\|S_{\pmb{\sigma}(n)}\dotsm S_{\pmb{\sigma}(1)}\|=+\infty$$
for any norm $\|\cdot\|$ on $\mathbb{R}^{4\times 4}$.
\end{remark}
\section{Concluding remarks}\label{sec6}
In this paper, we have introduced the dynamical concept---fiber-chaotic switching laws---for a discrete-time linear inclusion dynamical system that is induced by finitely many nonsingular square matrices.

We have proven that if the inclusion system has a stable word and meanwhile an expanding word, then its fiber-chaotic switching laws form a residual subset of its all possible switching laws (Theorem~\ref{thm1.2}). Therefore in this case, the ``generic'' dynamical behavior of this inclusion system is unpredictable and the state trajectories are infinitely sharply oscillated. In the two-dimensional case, if the inclusion system has an irrational rotation subsystem, an asymptotically divergent trajectory and an asymptotically stable trajectory, then it exhibits generic fiber-chaotic characteristic (Theorem~\ref{thm4.7}).

From Theorem~\ref{thm5.2} we see that if System (\ref{eq1.1}) has a fiber-chaotic switching law, then there coexist positive, zero and negative Lyapunov exponents. In addition, we have shown that if the inclusion system has the joint spectral (co-) radius $1$, then it does not have any fiber-chaotic switching laws (Theorem~\ref{thm5.1}). So a periodically stable system is simple from the viewpoint of fiber chaos (Theorem~\ref{thm5.3}).

For now, we will close with some questions that we are interested in.

\subsection{The inverse system}
Let $S_1,\dotsc,S_K$ all be nonsingular. Corresponding to System (\ref{eq1.1}), we consider its inverse system:
\begin{equation*}
y_n\in\{S_k^{-1}y_0\}_{k\in\bK},\quad y_0\in\mathbb{R}^d\textrm{ and }n\ge1.\eqno{(\ref{eq1.1})^*}
\end{equation*}
Since for any nonsingular square matrix $A$ it holds that $\|A^{-1}\|=\|A\|_{\mathrm{co}}^{-1}$ and $\|A^{-1}\|_{\mathrm{co}}=\|A\|^{-1}$, if System~(\ref{eq1.1}) satisfies the condition of Theorem~\ref{thm1.2} then System~$(\ref{eq1.1})^*$ is also generically fiber-chaotic. From Lemma~\ref{lem3.8} (more precisely the proof of Lemma~\ref{lem3.8}) and Theorem~\ref{thm5.2}, we can see that under the same hypotheses of Theorem~\ref{thm1.3}, System~$(\ref{eq1.1})^*$ is densely fiber-chaotic. Closely related to the above question, we now ask another question:

\begin{que}\label{q1}
If System~(\ref{eq1.1}) is generically/densely fiber-chaotic, is System~$(\ref{eq1.1})^*$ generically/densely fiber-chaotic?
\end{que}

If $d=2$ and System~(\ref{eq1.1}) is fiber-chaotic having an irrational rotation subsystem, then from Theorem~\ref{thm5.2} it follows that System~$(\ref{eq1.1})^*$ is fiber-chaotic. However, it is not a necessary condition for fiber chaos to have an irrational rotation subsystem, as was shown by Theorem~\ref{thm1.2}. So Question~\ref{q1} is nontrivial.

\subsection{Pointwise fiber-chaos}
We say that System (\ref{eq1.1}) is \textit{pointwise fiber-chaotic}, provided that for any nonzero initial state $x_0\in\mathbb{R}^d$, there corresponds some switching law $\sigma=\sigma_{x_0}\in\varSigma_{\bK}^+$ such that
\begin{equation*}
\liminf_{n\to\infty}\|x_n(x_0,\sigma)\|=0\quad\textrm{and}\quad\limsup_{n\to\infty}\|x_n(x_0,\sigma)\|=\infty.
\end{equation*}
Then by a slight improvement of the proof of Theorem~\ref{thm5.2}.\,(3), we see that a pointwise fiber-chaotic System~(\ref{eq1.1}) has simultaneously positive and  negative Lyapunov exponents. 

We now conclude this paper with this question:
\begin{que}
If System (\ref{eq1.1}) is pointwise fiber-chaotic, does there exist at least one switching law $\sigma\in\varSigma_{\bK}^+$ which is fiber-chaotic in the sense of Definition~\ref{def1.1}?
\end{que}

\end{document}